\documentclass[amsmath,amssymb,aps,longbibliography]{revtex4-1}

\usepackage{graphicx}
\usepackage{dcolumn}
\usepackage{graphicx}
\usepackage{dcolumn}
\usepackage{bm}

\usepackage{bm,color,graphicx,braket,mathptmx}

\usepackage{bbding}
\usepackage{graphicx}
\usepackage{amsmath,amsthm,amssymb,bm}
\usepackage{mathrsfs}
\usepackage{braket}
\usepackage{dsfont}
\usepackage[normalem]{ulem}
\usepackage{soul}
\usepackage{todonotes}
\usepackage{youngtab}
\usepackage{ytableau}
\usepackage{tikz}
\usepackage{pstricks}
\usepackage{tikz-cd}
\usepackage{aurical}
\usepackage{float}
\usepackage[T1]{fontenc}
\def\rsymbol{\large\hbox{\textbf{r}}}

\def\Msymbol{\hbox{\textbf{M}}}
\def\Asymbol{\hbox{\textbf{A}}}
\def\rsymbol{\hbox{\textbf{r}}}

\usepackage{xcolor}

\usepackage{dcolumn}

\usepackage{mathtools}
\DeclarePairedDelimiter{\ceil}{\lceil}{\rceil}

\newcounter{noitem}

\newenvironment{definition}[1][]{\refstepcounter{noitem}\par\smallskip
   \noindent \textbf{Definition~\thenoitem:\,#1} \rmfamily}{\smallskip}
   
\newenvironment{theorem}[1][]{\refstepcounter{noitem}\par\smallskip
   \noindent \textbf{Theorem~\thenoitem:\,#1} \rmfamily}{\smallskip}
   
\newenvironment{lemma}[1][]{\refstepcounter{noitem}\par\smallskip
   \noindent \textbf{Lemma~\thenoitem: #1} \rmfamily}{\smallskip}
   
 \renewenvironment{proof}[1][]{\par\smallskip
   \noindent \textbf{Proof:\,#1  }  \rmfamily }{\hfill $\square$ \smallskip}
   
   \newenvironment{example}[1][]{\refstepcounter{noitem}\par\smallskip
   \noindent \textbf{Example~\thenoitem:\,#1} \rmfamily}{\smallskip}
   
      \newenvironment{problem}[1][]{\refstepcounter{noitem}\par\smallskip
   \noindent \textbf{Problem~\thenoitem: #1} \rmfamily}{\smallskip}

  \newenvironment{proposition}[1][]{\refstepcounter{noitem}\par\smallskip
   \noindent \textbf{Proposition~\thenoitem:\,#1} \rmfamily}{\smallskip}
   
     \newenvironment{corollary}[1][]{\refstepcounter{noitem}\par\smallskip
   \noindent \textbf{Corollary~\thenoitem:\,#1} \rmfamily}{\smallskip}

\colorlet{shadecolor}{gray!80}

 \newcommand\dout{\!\!\!
        \bgroup
        \markoverwith{%
        \textcolor{red}%
        {
            \rule[0.2ex]{0.1pt}{1.0pt}%
            \hskip-0.1pt
            \rule[0.8ex]{0.1pt}{1.0pt}%
        }}
        \ULon
     }
     \MakeRobust\dout

\begin{document}
\setstcolor{red}

\title{Generalized Interference of Fermions and Bosons}

\author{Dylan Spivak}
\address{ Department of Mathematical Sciences, Lakehead University, Thunder Bay, Ontario P7B 5E1, Canada}

\author{Murphy Yuezhen Niu}
\address{ Department of Physics, Massachusetts Institute of Technology, Cambridge, Massachusetts 02139, USA }

\author{Barry C.\ Sanders}
\address{ Institute for Quantum Science and Technology,
        University of Calgary, Alberta T3A~0E1, Canada}
        
\author{Hubert de~Guise}
\address{Department of Physics, Lakehead University, Thunder Bay, Ontario P7B 5E1, Canada}


\begin{abstract} Using tools from representation theory, we derive expressions for the coincidence rate of partially-distinguishable particles in an interferometry experiment. Our expressions are valid for either bosons or fermions, and for any number of particles. In an experiment with $n$ particles the expressions we derive contain a term for each partition of the integer $n$; Gamas's theorem is used to determine which of these terms are automatically zero based on the pairwise level of distinguishability between particles. Most sampling schemes (such as Boson Sampling) are limited to completely indistinguishable particles; our work aids in the understanding of systems where an arbitrary level of distinguishability is permitted. As an application of our work we introduce a sampling scheme with partially-distinguishable fermions, which we call Generalized Fermion Sampling.
\end{abstract}

\maketitle

\section{Introduction}\label{intro}

In this paper we investigate coincidence rates for $n$ particles arriving not necessarily simultaneously at $n$ detectors located at the output of an $m\times m$ unitary linear interferometer. We extend previous work on coincidence rates \cite{TGdGS13,dGTPS14,TTS+15,dGS17,KSSdG18} for partially distinguishable particles, either bosons or fermions; in doing this we improve on the understanding of the Hong-Ou-Mandel effect~\cite{HongOuMandel} for many-particle systems~\cite{GenHOM}. Our work is motivated by the problem of \textsc{BosonSampling}~\cite{AA11}, which has thrown new bridges 
between computational complexity and linear optics but initially dealt with indistinguishable, simultaneous bosons. The boson-sampling computational challenge is to ascertain, computationally and experimentally, the distribution of coincidence rates as the location of the detectors is changed.
The excitement of boson sampling and its associated computational problem 
is that coincidence rates for indistinguishable bosons are given by permanents~\cite{marcus1965permanents} of (nearly) Gaussian-random $n\times n$ complex matrices;
whereas these permanents are hard to compute, they are easily accessible experimentally~\cite{KHHK19}.

Boson sampling has been generalized to non-simultaneous arrival times~\cite{TTS+15} and Gaussian states~\cite{LLR+14}. A more recent refinement of the original boson sampling proposal is the rapidly developing study and implementation of Gaussian BosonSampling ~\cite{hamilton2017gaussian,huh2017vibronic,arrazola2018quantum,kruse2019detailed,zhong2021phase}, which leverages the transformation properties of squeezed states to produce an equally hard computational task via computation of a hafnian rather than a permanent.
Originally regarded as the quickest way to establish quantum advantage~\cite{Pre12},
some argue that efficiently solving \textsc{BosonSampling} could yield a practical benefit~\cite{huh2015boson,nikolopoulos2016decision, chakhmakhchyan2017quantum, nikolopoulos2019cryptographic, banchi2020molecular}.

Whereas the interference of light is a venerable centuries-old topic~\cite{You02},
fermionic interference is relatively recent,
dating back to electron diffraction~\cite{DG28}.
The notion of fermion interferometry,
as a concept for one or more fermions to face one or more paths,
controlled by analogues of beam splitters and phase shifters,
extends the concept of fermion interference to controlled interferometric processes~\cite{Yur88}
such as electronic analogues to Young's double-slit experiment~\cite{BPLB13}
and fermion anti-bunching evident in current fluctuations of partitioned electrons~\cite{OKLY99}.
The term fermion interferometry has been used in nuclear physics~\cite{GC00},
based on a Hanbury Brown and Twiss type of two-nucleon correlation measurement~\cite{Koo77,PT87}
akin to fermion antibunching~\cite{OKLY99},
and not fermion interferometry as we study here. Some recent advances in the coherent control of electrons, including applications
to fermionic interferometry, are reviewed in \cite{bauerle2018coherent}.

It is convenient to think of the permanent, which is often
introduced as an ``unsigned determinant",  as a group
function on a $GL_n(\mathbb{C})$ matrix associated with the fully symmetric representation of the permutation group $\mathfrak{S}_n$ of $n$ objects. From this perspective, the determinant is also a group function but associated with the alternating representation of the permutation group
$\mathfrak{S}_n$.

It is well established that the permanent and the determinant are required to evaluate rates for simultaneous bosons and fermions, respectively.  Here, we will show as a first result of this paper that, for nonsimultaneous arrival times, additional group functions beyond the permanent and determinant and associated with other representations of $\mathfrak{S}_n$,  contribute to the 
coincidence rate.  The rate can also be expressed in terms of immanants~\cite{littlewood1934group},
which generalize matrix determinants and permanents and in fact interpolate between them. A second important result of our work follows as a corollary to Gamas's theorem \cite{gamas,BergetGamas}: we will show how the mutual partial distinguishability of the particles determines which group functions will have non-zero contributions. To our knowledge, this is the first time Gamas's theorem has been applied to a physical system.
Our third result states that, if at most $\ceil*{n/2}$ fermions are indistinguishable, then the exact coincidence rate for fermions contains computationally expensive group functions which require a number of operations that grows exponentially in the number of fermions. As a fourth result, we
show that, for uniformly random arrival times, the probability of needing to evaluate these computationally expensive group functions differs from
$1$ by a quantity that goes to $0$ exponentially fast in the number of fermions.
This last result naturally leads us to introduce \textsc{GenFermionSampling} as the computational problem of sampling the distribution of coincidence rates generated by fermions arriving nonsimultaneously at the detectors. 

Before we proceed we need to clarify issues of semantics. 
By ``arrival time'', 
we are referring to particles whose temporal profile (wave-packet shape) is effectively localized in time; we define the arrival times as the delay between particle creation at the source and the time taken for the wavepacket to travel from the source to the interferometer \emph{output} port.
Operationally,
our arrival times are effectively controllable delays in the interferometer input channels.  If two or
more arrival times are the same, the relevant particles ``arrive simultaneously'' at the detectors; if arrival times
are different, the particles are ``nonsimultaneous''.
Even if the particles 
do not arrive simultaneously at the detectors, 
we will speak of coincidence rate in the sense of particles detected 
during a single run
by detectors at selected positions.  We imagine an operator triggering the release 
of particles and opening a detection-time window of some fixed time interval, long 
enough for particles to arrive  at the interferometer, to be 
scattered in the interferometer and  counted at the detectors.  
The detection-time window
then closes and all the detections during this window are described as coincident 
detections.  The rate at which
they are detected in coincidence depends on the interferometer, the times of arrival
at the interferometer, and positions of the detectors; this rate is called the 
\emph{coincidence rate} throughout this paper.

\section{Notation and Mathematical Preliminaries}
\label{subsec:mathscattering}

\subsection{The scattering matrix and its submatrices}\label{subsec:interferometry}

We envisage an $m$-channel passive interferometer~\cite{yurke19862} that receives single particles at $n$ ($n\le m$) of its input ports and no particles at the rest of the input ports.
Mathematically,
this interferometer transformation is described by a unitary linear transformation $U$,
which is an $m\times m$ unitary matrix. We assume that $U$ is a Haar-random matrix
to avoid any symmetries in the matrix entries that would inadvertently simplify the rate calculations.

For an $n$-particle input state,
with one particle per input port,
we label without loss of generality this input by $1$ to $n$ as port labels,
or, in the Fock basis,
as $\ket{1}_n=\ket{1,1,1, \ldots, 0, 0, 0}$. Detectors are placed at the output ports, and,
although either one or zero particles enter each input,
more than one particle can be registered at an output detector.
For the generic case, we have $s_i$ particles detected at the $i^\text{th}$ channel, and the total number of particles is conserved such that $\sum_i s_i = n$.

Let $s\in \Phi_{m,n}$, where $\Phi_{m,n}$ is the set of $m$-tuples $(s_1,\ldots, s_m)$
so that $s_1+\ldots + s_m=n$.
Mathematically we first construct the $m \times n$ submatrix~$\mathcal{A}$
by keeping only the first $n$ columns of $U$.
Next, we construct the $n\times n$ submatrix $A(s)$ for each $s\in \Phi_{m,n}$, 
by replicating side-by-side $s_i$ times row $i$ 
of $\mathcal{A}$, and by deleting all rows for which~$s_i=0$. 

In the remainder of this paper
we focus on and restrict our presentation
and arguments to those strings $s$ with $s_i=0$ or $1$ for all $i$.
Practically, this can be achieved by diluting the particles so that 
$m\gg n$, in which case the probability that
$s$ will contain any $s_i>1$ is small. Alternatively, one may
post-select those events where the above condition holds, ignoring
events where some of the $s_i>1$.  
As a consequence, all denominator factors of $s_1!s_2! \ldots s_m!$ 
that would appear in a rate are $1$, and every submatrix $A(s)\in GL_n(\mathbb{C})$.

\subsection{Enter the permutation group}

As mentioned above we envisage a linear lossless interferometer modelled as a unitary $m\times m$ matrix. Suppose that~$n$ ($n< m$) particles enter in a subset of the possible input ports, exactly one particle per input port.  Without loss of generality we label these input ports from $1$ to $n$. The effect of the interferometer is to scatter each single particle so it eventually 
{reaches one of the detectors; each detector is located at a different output port, 
and we assume that the $n$ input particles reach $n$ different detectors with exactly one particle per detector.}
Given a string $s \in \Phi_{m,n}$ so that $s_i=0$ or $1$, the 
positions of the $n$ detectors registering one particle are now 
simply the list of those $i$'s for which $s_i=1$.  This list of 
detector positions 
we denote by $S$, so that
$S_k$ is the position of the $k^{th}$ detector that registers a particle. For example, if $s=(0,1,1,0,1)$, then we have $S=(2,3,5)$ with $S_1=2, S_2=3, S_3=5$
as shown in the example of Fig.~\ref{fig:scatterexample}.

Given $s \in \Phi_{m,n}$ and $S$, 
the $n\times n$ matrix $A(s)$ thus has elements
\begin{align}
A(s)_{\beta,\alpha}= U_{S_\beta,\alpha}\,  ,\qquad 
\alpha,\beta=1,\ldots,n\, .
\end{align}
With $s$ as above, for instance, we have, using $A(235):=A(01101)$
\begin{align}
    A(235)=
    \left(\begin{array}{ccc}
    A(235)_{11}&A(235)_{12}&A(235)_{13}\\
    A(235)_{21}&A(235)_{22}&A(235)_{23}\\
    A(235)_{31}&A(235)_{32}&A(235)_{33}
    \end{array}\right)= \left(\begin{array}{ccc}
    U_{21}&U_{22}&U_{23}\\
    U_{31}&U_{32}&U_{33}\\
    U_{51}&U_{52}&U_{53}
    \end{array}\right)\, .
\end{align}
The matrix $A(s)$ is in general not unitary but rather $A(s)\in GL_n(\mathbb{C})$.
(Note that, if we allowed more than one particle per output port, 
$A(s)$ is no longer in $GL_n(\mathbb{C})$.)

{During a run of the experiment,} the effect of the interferometer is to 
scatter a single particle state from $i$ to any \emph{one} output 
port $S_k$. {The data of the $n$ detectors that have registered a particle is then collected.} Fig.~\ref{fig:scatterexample}
also shows that
particle $1$ ends up in detector $2$, located at output port 
$3$; particle $2$ reaches detector $1$ at port $2$; and particle 
$3$ reaches detector $3$ at port $5$ .  
 The amplitude for this process is proportional to the monomial
$ A(s)_{21}A(s)_{12}A(s)_{33}$.
There are clearly six possible ways for 
particles to exit in the three detectors, with exactly one particle per detector, and each of these six possible ways is a permutation of $(S_1,S_2,S_3)$.  The net amplitude is thus a sum of $6$ terms of the form 
$A(s)_{k1}A(s)_{j2}A(s)_{\ell3}f(\bar\tau_1)g(\bar\tau_2)h(\bar\tau_3) $, 
where $(k,j,\ell)$ is a permutation of $(S_1,S_2,S_3)$
and $f,g,h$ are functions of the individual arrival times at the interferometer.

\begin{figure}[h!]
    \centering
    \includegraphics[scale=0.5]{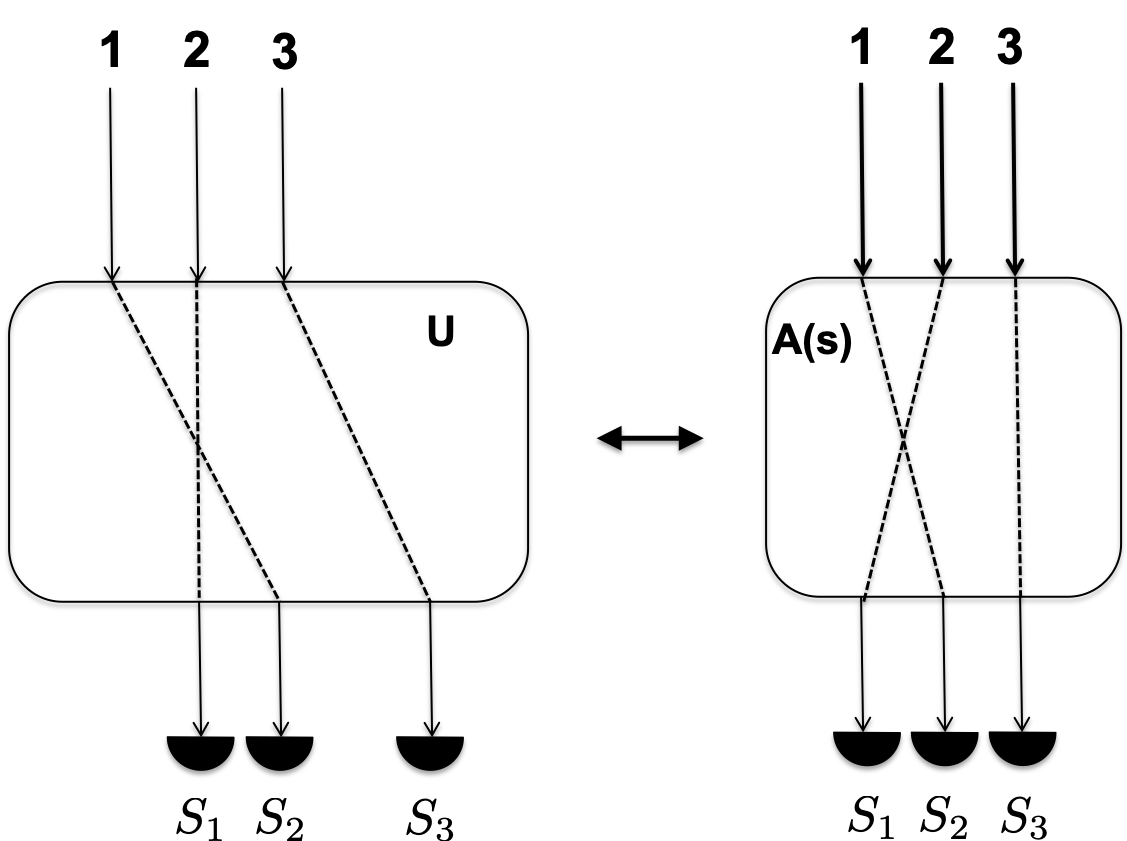}
    \caption{An example of an interferometer with $m=5$ and $3$ particles, illustrating
    the labelling and notation for the matrices $U$ and $A(s)$.}
    \label{fig:scatterexample}
\end{figure}

More generally, as we restrict ourselves to situations where the~$n$ particles are detected in different output ports, the effect of the interferometer is to shuffle the input $(1,\ldots,n)$ so that particles can exit in $n!$ possible ways at the output, with each of these indexed by a permutation of $(1,\ldots,n)$.  Each of these $n!$ ways is one term to be added to obtain the net scattering amplitude for this fixed set of detector positions. We see that the permutation (or symmetric) group $\mathfrak{S}_n$ ~\cite{james2009representation,tung1985group} of~$n$ particles enters into the problem of computing coincidence rates in a natural way. 

\subsection{Partitions, Immanants, and group functions}
\label{subsec:partitions}

Here we explain matrix immanants which - like the permanent or the determinant - are polynomials in the entries of a matrix, albeit neither fully symmetric nor antisymmetric, and are thus naturally well suited 
to understand the structure of coincidence rates for partially distinguishable particles.  Immanants are closely related to Young diagrams, Young tableaux, and general irreducible representations of $\mathfrak{S}_n$; the references~\cite{lichtenberg2012unitary,rowe2010fundamentals,Yon07} are especially helpful for a 
longer discussion of these concepts.  Additional details on this background can be also found in \cite{spivak2020immanants}.

The study of the permutation group is closely tied to the partitions of $n$, as these partitions label the irreducible representations of $\mathfrak{S}_n$.
A partition of a positive integer~$n$ is a $k$-tuple $\lambda = (\lambda_1, \lambda_2, \dots, \lambda_k)$ such that $\lambda_1 \geq \lambda_2 \geq \dots \geq \lambda_k \geq 1$ and $\sum_{i=1}^k \lambda_i=n$; we call $k$ the number of parts of a partition, $\lambda_1$ the width, and~$n$ the size. The notation $\lambda \vdash n$ means that $\lambda$ is a partition of $n$. A Young diagram is a visual way to write a partition using boxes.
Given a partition $\lambda$, we construct the associated Young diagram by placing a row of $\lambda_1$ boxes at the top, then we add a left-justified row of $\lambda_2$ boxes is added directly below. This process is repeated for each $\lambda_i$. The partition $\lambda$ is referred to as the shape of the Young diagram. 

We  also need the partition $\lambda^*$ \emph{conjugate} to the partition $\lambda$, 
which is obtained from $\lambda$ by exchanging rows and columns of $\lambda$.  This is equivalent 
to reflecting the diagram $\lambda$ about the main diagonal.
\begin{example}
Conjugate partitions for $\mathfrak{S}_4$:
\begin{align}
\Yvcentermath1\Yboxdim{8pt}\yng(4)\hspace{4mm} &\hbox{is conjugate to} \hspace{4mm} \Yvcentermath1\Yboxdim{8pt}\yng(1,1,1,1) \nonumber \\
\Yvcentermath1\Yboxdim{8pt}\yng(3,1)\hspace{4mm} &\hbox{is conjugate to} \hspace{4mm} \Yvcentermath1\Yboxdim{8pt}\yng(2,1,1)  \nonumber \\
\Yvcentermath1\Yboxdim{8pt}\yng(2,2)\hspace{4mm} &\hbox{is conjugate to itself}\ \nonumber
\end{align}
\end{example}

A Young tableau (plural tableaux) on a Young diagram of size~$n$ is a numbering of the boxes using entries from the set $\{1,2,\dots,n\}$. A standard tableau has its entries strictly increasing across rows and strictly increasing down the columns, as a result each integer from $1$ to~$n$ appears exactly once in the tableau. The conditions for the semi-standard tableaux are relaxed: entries weakly increase along rows, but still must increase strictly down columns. We use the notation 
$\text{s}_\lambda$ and $\text{d}_\lambda$ for the number of standard and semi-standard Young tableaux of shape $\lambda$, respectively. 

For an arbitrary $n \times n$ matrix $B$, 
there is an immanant defined for every partition of $n$. The $\lambda$-immanant of $B$ is 
\begin{align}
\operatorname{imm}^\lambda{(B)}:=\sum_{\sigma \in \mathfrak{S}_n} \chi_\lambda(\sigma) \prod_{i=1}^n B_{\sigma(i),i}.
\end{align}
The notation $\chi_\lambda(\sigma)$ refers to the 
character of element~$\sigma$ in the $\lambda$-representation of $\mathfrak{S}_n$. The permanent 
and determinant are special cases of immanants that correspond to 
the trivial representation and the alternating representation, 
respectively. The simplest non-trivial example is the 
$\Yvcentermath1\Yboxdim{4pt}\yng(2,1)$-immanant of a $3 \times 3$ 
matrix 
\begin{align}
\operatorname{imm}^{\Yvcentermath1\Yboxdim{4pt}\yng(2,1)}{(B)} = 2 B_{11}B_{22}B_{33} - B_{21}B_{32}B_{13} - B_{31}B_{12}B_{23}.   
\end{align}

Kostant provides a connection between certain group functions of $GL_n(\mathbb{C})$ and immanants~\cite{kostant}.
Let $B\in GL_n(\mathbb{C})$ and consider the representation $D_\lambda$ of  $GL_n(\mathbb{C})$,
namely,
\begin{align}
D_\lambda: GL_n(\mathbb{C}) &\rightarrow GL_{\text{d}_\lambda}(\mathbb{C})\\
            B &\rightarrow D_\lambda(B). \nonumber
\end{align}
A $\mathcal{D}$-function is a function of the form 
\begin{align}
\mathcal{D}^\lambda_{i,j}(B) = \bra{\psi^\lambda_i} D_\lambda(B) \ket{\psi^\lambda_j},  \label{eq: innerprod}
\end{align}
where $\ket{\psi_i^{\lambda}}$ is a basis vector of the representation space of 
the $D_\lambda$ representation of $GL_n(\mathbb{C})$. Define 
$\mathcal{D}^{\lambda}(B)$ to be the matrix with $\mathcal{D}^{\lambda}_{i,j}(B)$ 
as its $ij^\text{th}$ entry. We are particularly interested in the states 
$\ket{\psi_j^\lambda}$ that live in the $(1,1,1,...,1)$-weight subspace of 
$D_\lambda$ (if $B\in SU(n)$ this is referred to as the $0$-weight subspace).
Recall that a vector $v\in \mathbb{C}^n$ has weight $w=(w_1,w_2,\ldots,w_n)$ when
\begin{align}
        \hbox{diag}(t_1,t_2,\ldots,t_n)v= t_1^{w_1} t_2^{w_2}\ldots t_n^{w_n}v.
\end{align}
The following theorem of Kostant shows a deep relationship between immanants and $\mathcal{D}$-functions.\\
\begin{theorem}
\label{theorem:Kostant}
(Kostant~\cite{kostant}): Given $B \in GL_n(\mathbb{C})$, it follows that
\begin{align}
\operatorname{imm}^\lambda(B) = \sum_{i} \mathcal{D}^\lambda_{i,i}(B),   \label{eq:immassumofD} 
\end{align}
where the sum is over all states of weight $(1,1,1,...,1)$.
\end{theorem}
\noindent
This connection is important to establish the complexity of evaluating the coincidence rate equations.

\section{Rates for continuous time}\label{sec:approach}

We analyze coincidence rates of an experiment where~$n$ particles are modelled by a Gaussian wave packet arriving in an $m\times m$  interferometer at times $\bar\tau_1,\ldots,\bar\tau_n$,
and obtain a compact form for the bosonic and fermonic coincidence rate equations.
The difference between the boson and fermion cases is a sign of some entries of an $n!\times n!$ symmetric matrix $R(\bar{\bm\tau};\text{p})$ 
containing information about the level of distinguishability through the arrival times $\bar\tau$. The species label $\text{p}$ can be either $\text{b}$ for bosons or $\text{f}$ for fermions. The permutation group $\mathfrak{S}_n$ of~$n$ particles is an important tool throughout this discussion. If two or more elements of $\mathfrak{S}_n$ need to be specified, we do so with a subscript \emph{v.g.} $\sigma_i$ or $\gamma_j$ \emph{etc}. 

\subsection{Sources, input state, arrival-time vector, detectors}
\label{subsec:source}

We model a source for input channel $k$ as producing 
a temporal Gaussian wave packet in channel $k$; the wavepacket is described by the one-particle state
\begin{align}
\ket{k(\bar\tau_k)}:= \hat{A}^\dagger_k(\bar{\tau}_k)\ket0  = \int\text{d}\omega_k\,\phi(\omega_k) \text{e}^{-\text{i}\omega_k \bar{\tau}_k} \hat{a}_k^\dagger(\omega_k)\ket0,
\end{align}
where the spectral profile
\begin{align}
\phi(\omega_k) = \frac1{\sqrt{\Delta \omega\sqrt{2 \pi}}} \text{e}^{-\frac{(\omega_k - \omega_0)^2}{4(\Delta\omega)^2}} \label{eq:spectralprofile}
\end{align}
is normalized so that $\int\text{d}\omega_k \left \vert \phi(\omega_k) \right \vert^2=1$.
The frequency spread of the source $\Delta \omega$ is sufficiently narrow and assumed constant, 
so that one can replace
the energy $\hbar \omega$ of each Fourier component by the mean excitation energy of the photon \cite{fearn1989theory}. 
Here, 
operators such as $\hat{a}_i^\dagger(\omega_k)$ ---they can be either fermionic or bosonic---create an excitation of frequency $\omega_k$ in input channel $i$. 

We assume~$n$ independent but identical sources, each delivering one wave packet 
to a separate input channel,
with~$n$ smaller than or equal to the total number~$m$ of channels.
Let
\begin{align}
\ket0 = \ket0_1 \otimes \ket0_2 \otimes \ket0_3 \otimes  \dots \otimes \ket0_n,
\end{align}
denote the $n$-particle vacuum state; the $n$-particle input state is 
then the $n$-fold product state~\cite{fearn1989theory}
\begin{align}
\ket{\psi_{in}(\bar{\tau}_1,\ldots,\bar{\tau}_n)}:= 
\hat{A}^\dagger_1(\bar{\tau}_1)  \hat{A}^\dagger_2(\bar{\tau}_2)  \dots \hat{A}^\dagger_n(\bar{\tau}_n)\ket0\, .
\label{eq:continuousinput}
\end{align}
The times of arrival
\begin{equation}
\label{eq:arrivaltimevector}
\bar{\bm\tau}:=[\bar{\tau}_1,\ldots,\bar{\tau}_n]
\end{equation}
that appear in Eq.~(\ref{eq:continuousinput}) are continuous real numbers in some finite time interval of duration $\mathcal{T}$ and 
are measured from some common but otherwise arbitrary reference time.  
The detectors {are Fourier-limited \cite{santori2002indistinguishable}.  They} integrate with frequency response {$\Phi(\Delta\Omega)$} over the entire interval $\mathcal{T}$ 
{so that long integration times lead to small width $\Delta \Omega$}; they are modelled as projection operators 
\begin{align}
    \hat \Pi_k=&\int\text{d}\Omega_k \hat{a}_k^\dagger(\Omega_k)\ket0\bra0 \hat{a}_k(\Omega_k) 
    {\Phi(\Delta\Omega_k)}\, , \\
    \hat \Pi=&\hat \Pi_1\otimes \hat \Pi_2\otimes \cdots \otimes \hat \Pi_n. 
    \label{eq:detectorprojector}
\end{align}
For simplicity we choose the frequency response of all detectors to be identical Gaussians of
the form
\begin{align}
    \Phi(\Delta\Omega)=\frac1{\Delta \Omega\sqrt{2 \pi}} \text{e}^{-\frac{(\Omega - \omega_0)^2}{2(\Delta\Omega)^2}} \label{eq:frequencyresponse}
\end{align}
with fixed width $\Delta\Omega$.  A detector does not register arrival times, but registers only if a
particle has entered the detector at any time during one run of the experiment.
The probability of detection at a particular detector 
depends on $\bar{\bm\tau}$ and on the submatrix $A(s)$.


\subsection{Exact coincidence rate equations, continuous time}\label{sec:exactrates}

In this subsection, we obtain a compact form for the bosonic and fermionic
coincidence rate equations, and show how they are closely related. 
We also show that the entries of  an $n!\times n!$ symmetric matrix $R(\bar{\bm\tau};\text{p})$ are monomials in the entries of the so-called
delay matrix $r(\bar{\bm\tau})$, defined in Eq.~(\ref{eq:ratematrix}).
This observation is crucial to our simplifications. In \S\ref{sec:results}, we
use Gamas's theorem to deduce the conditions under which the $\lambda$-immanant of $r(\bar{\bm\tau})$ is
zero based on the set of arrival times $\bar{\bm\tau}$.    

The effect of the interferometer is to scatter each single particle state to  
 \begin{align}
      \ket{k(\bar\tau_i)}\to A(s)\ket{k(\bar\tau_i)} = \sum_{q} \ket{q(\bar\tau_i)} A(s)_{qk}\, ,
 \end{align}
\emph{i.e.} the scattering matrix produces superposition of states from the initial state but does not affect the times of arrival.  The scattered state can be written in expanded form as
 \begin{align}
 \ket{\hbox{out}(\bar{\bm{\tau}})}=&
   \sum_{\sigma\in \mathfrak{S}_n} A(s)_{\sigma(1)1}A(s)_{\sigma(2)2}\ldots A(s)_{\sigma(n)n} \nonumber \\
   &\times 
\int\text{d}\omega_1 \text{d}\omega_2 \cdots \text{d}\omega_n \phi(\omega_1)\phi(\omega_2)\cdots \phi(\omega_n)
\text{e}^{-i\omega_1\bar\tau_1}\text{e}^{-i\omega_2\bar\tau_2}\cdots \text{e}^{-i\omega_n\bar\tau_n} \nonumber \\
&\qquad\qquad \qquad\qquad   \times 
\hat a^\dagger_{\sigma(1)}(\omega_1)\hat a^\dagger_{\sigma(2)}(\omega_2)
\cdots \hat a^\dagger_{\sigma(n)}(\omega_n)\ket0\, .
\end{align}
To simplify the expressions for the rates, choose some ordering for the elements
$\gamma_i$ of $\mathfrak{S}_n$ and introduce the following shorthand notation for
monomials. Let $M$ be some $n \times n$ matrix, and write for the product
\begin{align}
M_{\gamma(1),1}M_{\gamma(2),2} \dots M_{\gamma(n),n}
:= \Msymbol_{\gamma}\, . \label{eq:shorthandnotation}
\end{align}
To highlight and contrast the boson and fermion cases, we treat them in quick
succession.  In the bosonic case, the evaluation of the coincidence rate
\begin{align}
\operatorname{rate}(\bar{\bm\tau},s;\text{b})&:=\bra{\hbox{out}}\hat \Pi\ket{\hbox{out}}\, ,
\label{eq:barerates}
\end{align}
is best understood when written in the form
\begin{align}
\bra{\hbox{out}} \hat \Pi\ket{\hbox{out}} &= 
{{\cal N}^{n!}}\sum_{\sigma,\gamma \in \mathfrak{S}_n} \Asymbol(s)^*_{\gamma} 
[R(\bar{\bm\tau};\text{b})]_{\gamma,\sigma}
\Asymbol(s)_{\sigma} 
\label{eq: brate1} \, ,\\
{\cal N}&={\frac{1}{\sqrt{2\pi (\Delta\omega^2+\Delta \Omega^2)}}},
\end{align}
where {${\cal N}$ is a normalization factor}, $\Asymbol(s)_{\sigma}$ is the monomial defined in Eq.~(\ref{eq:shorthandnotation}), and
where the entries of the $n!\times n!$ matrix $R(\bar{\bm\tau};\text{b})$ 
(which we call the (bosonic) rate matrix) are evaluated by taking the operators $\hat a^\dagger_k(\omega_j)$ and $\hat a_i(\omega_\ell)$
to satisfy the usual boson commutation relations.  
Evaluation of the element $(\gamma,\sigma)$ of the matrix $R(\bar{\bm\tau};\text{b})$ yields
 \begin{align}
[R(\bar{\bm\tau};\text{b})]_{\gamma,\sigma} &=
\exp\left({-\frac{{\delta\omega^2}(\bar{\tau}_{\gamma^{-1}(1)}-\bar{\tau}_{\sigma^{-1}(1)})^2}{2}}\right)
 \exp\left({-\frac{{\delta\omega^2}(\bar{\tau}_{\gamma^{-1}(2)}-\bar{\tau}_{\sigma^{-1}(2)})^2}{2}} \right)  \nonumber \\
&\qquad \times  \dots \times \exp\left({-\frac{{\delta\omega^2}(\bar{\tau}_{\gamma^{-1}(n)}-\bar{\tau}_{\sigma^{-1}(n)})^2}{2}}\right)\, ,\\
{\frac1{\delta\omega^2}}&{=\frac1{\Delta \Omega^2}+\frac1{\Delta\omega^2}}.
 \end{align}
For fermions the operators $\hat a^\dagger_k(\omega_j)$ and
$\hat a_i(\omega_\ell)$ satisfy anticommutation relations so an additional negative sign, which comes from an odd number 
of permutations of fermionic operators, multiplies the $[R(\bar{\bm\tau};\text{b})]_{\gamma,\sigma}$ term
when the permutation $\sigma \gamma$ is odd.

{Ignoring the unimportant normalization constant ${\cal N}$ for convenience,} the coincidence rate is expressed as
\begin{align}
\operatorname{rate}(\bar{\bm\tau},s;\text{f})=& \sum_{\sigma,\gamma \in \mathfrak{S}_n} \Asymbol(s)_{\gamma}^* [R(\bar{\bm\tau};\text{f})]_{\gamma,\sigma}\Asymbol(s)_{\sigma}  \label{eq: frate1}  \\
[R(\bar{\bm\tau};\text{f})]_{\gamma,\sigma}=& \operatorname{sgn}(\sigma \gamma) 
[R(\bar{\bm\tau};\text{b})]_{\gamma,\sigma} =\chi^{(1^n)}(\sigma\gamma) [R(\bar{\bm\tau};\text{b})]_{\gamma,\sigma} \label{eq:fermionrate}
\end{align}
where $\chi^{(1^n)}(\sigma\gamma)$ is the character of element $\sigma\gamma$ in
the one-dimensional fully-antisymmetric (alternating) representation of
$\mathfrak{S}_n$.  Eq.~(\ref{eq:fermionrate}) shows explicitly how a sign
difference between the fermion and boson case can arise.  

We next analyze the rate matrices $R(\bar{\bm\tau};\text{b})$ and 
$R(\bar{\bm\tau};\text{f})$ by introducing the delay matrix~$r(\bar{\bm\tau})$, which is an $n \times n$ symmetric matrix that keeps track of the relative 
overlaps between pulses. The $(i,j)\text{th}$ entry of $r(\bar{\bm\tau})$ is
\begin{align}
r_{ij} := \text{e}^{\frac{-{\delta \omega^2}(\bar\tau_i-\bar\tau_j)^2}{2}}, 
\label{eq:delaymatrix}
\end{align}
The entries of this matrix {are bounded} from $0$ to $1$, and depend on the level of distinguishability between 
the particles. {In particular, $r_{ij}=1$ when the arrival times $\tau_i=\tau_j$}. 

We observe that the delay matrix 
$r(\bar{\bm\tau})$ is in fact a Gram 
matrix by considering the basis functions 
\begin{align}
f_k(\omega;\bar{\tau}_k) = \text{e}^{-\text{i}(\omega-\omega_0) \bar{\tau}_k}, \qquad  k=1,2,\dots,n  
\end{align}
with the symmetric inner product 
\begin{align}
r_{ij}(\bar{\bm\tau})=\bra{f_i} f_j \rangle = {\frac{1}{\cal N}}\int\text{d}\omega {\,|\phi(\omega)|^2 
\,\Phi(\Delta\omega)}
f_i^*(\omega;\bar{\tau}_i) f_j(\omega;\bar{\tau}_k), 
\label{eq:ratematrix}
\end{align}
with $\phi$ and $\Phi$ respectively given in Eqs.~(\ref{eq:spectralprofile}) and~(\ref{eq:frequencyresponse}).
Using the same shorthand notation for the monomials in $r$ as done
in Eq.~(\ref{eq:shorthandnotation}), we find that the $ij^\text{th}$ entry of the bosonic rate matrix $R(\bar{\bm\tau},\text{b})$  is
\begin{align}
R(\bar{\bm\tau};\text{b})_{ij}:= \rsymbol_{\gamma_j^{-1}\gamma_i}(\bar{\bm\tau}) = r_{\gamma_j^{-1}\gamma_i(1),1}(\bar{\bm\tau})\, r_{\gamma_j^{-1}\gamma_i(2),2}(\bar{\bm\tau}) \dots r_{\gamma_j^{-1}\gamma_i(n),n}(\bar{\bm\tau}).
\label{eq:ratematrixentry}
\end{align}
In the fermionic case, the $(i,j)\text{th}$ entry of the rate matrix is
\begin{align}
R(\bar{\bm\tau};\text{f})_{ij}  =  \operatorname{sgn}(\gamma_i \gamma_j)\rsymbol_{\gamma_j^{-1}\gamma_i}(\bar{\bm\tau}).
\end{align}

\noindent For some fixed ordering $\{\gamma_1, \gamma_2, \dots \gamma_{n!}\}$ of the elements of 
$\mathfrak{S}_n$, we 
{conveniently introduce}
the vector $v(s)$ 
\begin{align}
v(s)= (\Asymbol(s)_{\gamma_1}, \Asymbol(s)_{\gamma_2}, \dots , \Asymbol(s)_{\gamma_{n!}})^{\top}\,.
\end{align}
Eqs.~(\ref{eq: brate1}) and  (\ref{eq: frate1}) can respectively be
expressed as 
\begin{align}
\operatorname{rate}(\bar{\bm\tau},s;\text{b})
=& \sum_{\sigma,\gamma \in \mathfrak{S}_n} \Asymbol(s)_{\gamma}^* 
\rsymbol_{\gamma^{-1} \sigma}(\bar{\bm\tau}) \Asymbol(s)_{\sigma},  \label{eq:rateboson}\\
&= v^\dagger(s) R(\bar{\bm\tau};\text{b}) 
v(s) \, , \label{eq:rateuRu} 
\end{align}
and
\begin{align}
\operatorname{rate}(\bar{\bm\tau},s;\text{f})=
& \sum_{\sigma,\gamma \in \mathfrak{S}_n} 
\operatorname{sgn}(\gamma \sigma) \Asymbol(s)_{\gamma}^* \rsymbol_{\gamma^{-1} \sigma}(\bar{\bm\tau}) \Asymbol(s)_{\sigma}\, , \label{eq:ratefermion} \\
=& v^\dagger(s) R(\bar{\bm\tau};\text{f}) 
v(s). \label{eq:fermionrateuRu} 
\end{align}
We will refer to Eqs.~(\ref{eq:rateuRu}) and~(\ref{eq:fermionrateuRu}) as the 
coincidence rate equations for bosons and fermions, respectively. 

Consider for instance the case $n=3$.  The scattering matrix $A$ 
(dropping $s$ for the moment to avoid clutter) is 
\begin{align}
    A=\begin{pmatrix}
    A_{11}&A_{12}&A_{13}\\
     A_{21}&A_{22}&A_{23}\\
      A_{31}&A_{32}&A_{33}
    \end{pmatrix}.
\end{align}
We choose the ordering of $\mathfrak{S}_3$ elements
\begin{align}
    \{e, (12), (13), (23), (123), (132) \}
\end{align}
and 
the vector $v(s)$ takes the form 
\begin{align}
    v(s)=&
    \bigg(A_{11}A_{22}A_{33},\,A_{21}A_{12}A_{33},\,A_{31}A_{22}A_{13},\, A_{11}A_{32}A_{23},\,
    A_{21}A_{32}A_{13},\, A_{31}A_{12}A_{23} \bigg)^\top\, ,\nonumber \\
    =&\left(\Asymbol_e,
    \Asymbol_{(12)},
    \Asymbol_{(13)},
     \Asymbol_{(23)},
    \Asymbol_{(123)},
    \Asymbol_{(132)}\right)^\top\, .  \label{eq:Umonomials}
\end{align}
with $s$ implied in the argument of each $A_{ij}$.

The bosonic rate matrix $R(\bar{\bm\tau};\text{b})$ takes the form of a Schur power matrix \cite{powermatrix}:
\begin{align}
&R(\boldsymbol{\bar{\bm\tau};\text{b})}
=\left(
\begin{array}{cccccc}
 \rsymbol_e & \rsymbol_{(12)} & \rsymbol_{(13)} & \rsymbol_{(23)} & \rsymbol_{(123)} & \rsymbol_{(132)} \\
 \rsymbol_{(12)} & \rsymbol_e & \rsymbol_{(132)} & \rsymbol_{(123)} & \rsymbol_{(23)} & \rsymbol_{(13)} \\
 \rsymbol_{(13)} & \rsymbol_{(123)} & \rsymbol_e & \rsymbol_{(132)} & \rsymbol_{(12)} & \rsymbol_{(23)} \\
 \rsymbol_{(23)} & \rsymbol_{(132)} & \rsymbol_{(123)} & \rsymbol_e & \rsymbol_{(13)} & \rsymbol_{(12)} \\
 \rsymbol_{(132)} & \rsymbol_{(23)} & \rsymbol_{(12)} & \rsymbol_{(13)} & \rsymbol_e & \rsymbol_{(123)} \\
 \rsymbol_{(123)} & \rsymbol_{(13)} & \rsymbol_{(23)} & \rsymbol_{(12)} & \rsymbol_{(132)} & \rsymbol_e \\
\end{array}
\right)\, ,\\  \nonumber 
\end{align}
with the arrival-time vector $\bar{\bm\tau}$ implicit
in $\rsymbol$. For the fermionic case, we add a negative sign to the entries when $\gamma_i \gamma_j$ is an odd permutation; the result is 
\begin{align}
R(\boldsymbol{\bar{\bm\tau};\text{f})}
=&\left(
\begin{array}{cccccc}
 \rsymbol_e & -\rsymbol_{(12)} & -\rsymbol_{(13)} & -\rsymbol_{(23)} & \rsymbol_{(123)} & \rsymbol_{(132)} \\
 -\rsymbol_{(12)} & \rsymbol_e & \rsymbol_{(132)} & \rsymbol_{(123)} & -\rsymbol_{(23)} & -\rsymbol_{(13)} \\
 -\rsymbol_{(13)} & \rsymbol_{(123)} & \rsymbol_e & \rsymbol_{(132)} & -\rsymbol_{(12)} & -\rsymbol_{(23)} \\
 -\rsymbol_{(23)} & \rsymbol_{(132)} & \rsymbol_{(123)} & \rsymbol_e & -\rsymbol_{(13)} & -\rsymbol_{(12)} \\
 \rsymbol_{(132)} & -\rsymbol_{(23)} & -\rsymbol_{(12)} & -\rsymbol_{(13)} & \rsymbol_e & \rsymbol_{(123)} \\
 \rsymbol_{(123)} & -\rsymbol_{(13)} & -\rsymbol_{(23)} & -\rsymbol_{(12)} & \rsymbol_{(132)} & \rsymbol_e \\
\end{array}
\right).
\end{align}

Note that, as $r_{ij}(\bar{\bm\tau})=r_{ji}(\bar{\bm\tau})$, 
for any permutation~$\sigma$ we have that
$\textbf{r}_{\sigma}(\bar{\bm\tau})
=\textbf{r}_{\sigma^{-1}}(\bar{\bm\tau})$, 
which in turn means that the rate
matrix is symmetric with $r_{ii}=1$.
In Table \ref{table:rmonomials} each of the 
$\rsymbol_{\sigma}(\bar{\bm\tau})$
terms is expressed as a monomial in the entries of the delay matrix. 
\begin{table}[h!]
    \centering
        \begin{tabular}{|c|c|c|}
        \hline
        $\rsymbol_{\sigma}$ & Monomial in $r_{ij}$ entries & Simplified Monomial \\
           \hline
             $\rsymbol_{e}$ & $r_{11}r_{22}r_{33}$ & 1 \rule{0pt}{3ex} \\
            \hline
              $\rsymbol_{(12)}$ & $r_{21}r_{12}r_{33}$ & $r_{12}^2$ \rule{0pt}{3ex} \\
            \hline
             $\rsymbol_{(13)}$ & $r_{31}r_{22}r_{13}$ & $r_{13}^2$ \rule{0pt}{3ex} \\
   \hline
             $\rsymbol_{(23)}$  & $r_{11}r_{32}r_{23}$ & $r_{23}^2$ \rule{0pt}{3ex} \\
   \hline
    $\rsymbol_{(123)}$ & $r_{21}r_{32}r_{13}$ & $r_{12}r_{23}r_{13}$ \rule{0pt}{3ex} \\
   \hline
   $\rsymbol_{(132)}$ & $r_{31}r_{12}r_{23}$  & $r_{12}r_{23}r_{13}$ \rule{0pt}{3ex} \\
   \hline
        \end{tabular}
        \caption{Entries of the coincidence rate matrices. 
        The arrive-time vector $\bar{\bm\tau}$ is implicit.}
     \label{table:rmonomials}
    \end{table}

The expression of Eq.~(\ref{eq:rateuRu}) clearly shows that, in general, the exact expression for the rate involves the multiplication of $1\times n!\times n!\times n!\times n!\times 1$ quantities.  Our job is to cut down on this in two steps. Anticipating results of \S\ref{sec:results}, we make the crucial observation that, {if none of the $r_{ij}$ terms are $0$,} every permutation occurs exactly once in each row and each column of the rate matrix $R(\bar{\bm\tau};\text{p})$, hence $R(\bar{\bm\tau};\text{p})$ carries the regular representation $\rho_\text{reg}$ of the permutation group. (Although we are here working under the assumption
that $s_i=0$ or $1$, this observation on the structure of the rate matrix holds even if some $s_i>1$.)  It is well known that in the decomposition of the regular representation, each irrep of $\mathfrak{S}_n$ appears as many times as its dimension 
\begin{align}
\rho_\text{reg} = \bigoplus_{\lambda \vdash n} \bigoplus^{\text{s}_\lambda} \rho_\lambda.
\end{align}
It follows there exists a linear transformation $T$,
{determined by the representation theory of $\mathfrak{S}_n$,} that brings the rate matrix to its block diagonal form. In \S\ref{subsec:block}, we use this property to simplify the coincidence rate equations. Secondly, we 
show in \S\ref{subsec:Gamas} that, for a given set of arrival times $\bar{\bm\tau}$, certain blocks of the diagonalized rate matrix are automatically zero as a consequence of Gamas's theorem.

\section{Block diagonalization and continuous time}
\label{sec:results}
 
We proceed by leveraging the permutation symmetries of the matrix 
$R(\bar{\boldsymbol{\tau}};\text{p})$ and decompose this matrix into irreps of $\mathfrak{S}_n$, \emph{i.e.} we now show how it is possible to transform the expression of the coincidence rate equations
to one where the rate matrix $R(\bar{\bm{\tau}};\text{b})$ or 
$R(\bar{\bm{\tau}};\text{f)}$ has a block-diagonal
form. These symmetries in turn stem from the structure of the matrix $r(\bar{\boldsymbol{\tau}})$. We assume in \S\ref{subsec:SW}, \S\ref{subsec:block} and \S\ref{subsec:immanants}, that 
for particles $i$ and $j$ with $i\ne j$, no  $r_{ij}=0$, \emph{i.e.} no particle is fully distinguishable from any of the others. When some $r_{ij}=0$, the matrix $R(\bar{\bm{\tau}};\text{f)}$ will have $0$'s and 
may carry representations of $\mathfrak{S}_{n-p}$, so the problem is effectively one of $n-p$ partially distinguishable particles. We will discuss the case where some 
{of the} particles are fully distinguishable from others in  \S\ref{subsec:distinguishable}.

An essential point is that, when all $r_{ij}\ne 0$, the block-diagonalization procedure is independent of the arrival times and independent of the numerical values of entries in the submatrix $A(s)$: it depends \emph{only} on the 
action of the permutation group $\mathfrak{S}_n$ on monomials $\textbf{A}(s)_\gamma$ and $\mathbf{r}_\sigma(\bar{\boldsymbol{\tau}})$.

\subsection{Schur-Weyl duality and the decomposition of $\otimes^n\mathbb{C}^n$}
\label{subsec:SW}

We first show that for any string $s$ with $s_i=0$ or $1$, an input state of 
the type Eq.~(\ref{eq:continuousinput}) can be decomposed into pieces which
transform nicely under the action of the permutation group. 
For fixed $\bar{\bm\tau}$, the states
$\{\hat{A}^\dagger_j(\bar{\bm\tau})\ket0, j=1,\ldots,n\}$ form a basis
in $\mathbb{C}^n$ for the $n$-dimensional defining representation
of $GL_n(\mathbb{C})$:
    \begin{align}
        \hat{A}^\dagger_1(\bar{\bm\tau})\ket0\to \left(\begin{array}{c}1\\0 \\ \vdots\\0\end{array}\right)\, ,\quad 
         \hat{A}^\dagger_2(\bar{\bm\tau})\ket0\to \left(\begin{array}{c}0\\1 \\ \vdots\\0\end{array}\right)\, ,\quad 
         \ldots \quad 
          \hat{A}^\dagger_n(\bar{\bm\tau})\ket0\to \left(\begin{array}{c}0 \\ 0\\\vdots\\1\end{array}\right)\, .
    \end{align}
 Clearly $\hat{A}^\dagger_j(\bar{\bm\tau})\ket0$ has weight $(0,0,\ldots,1_j\ldots,0)$, with $0$ everywhere
    except $1$ in the $j^\text{th}$ entry. Thus, the product state of Eq.~(\ref{eq:continuousinput}) is a state of weight 
$(1,1,\ldots,1)$ in the $n$-fold
tensor product of defining representation of $GL_n(\mathbb{C})$.  
Scattering by a $GL_n(\mathbb{C})$ matrix and subsequent arrivals at different detectors map this product state to a linear combination of states in the same space of states with weight $(1,1,\ldots,1)$.

Now, it is well known that this $n$-fold 
tensor product decomposes as 
\begin{align}
    \otimes^n \mathbb{C}^n
        =\sum_{\lambda \vdash n} \rho_\lambda \otimes D_\lambda 
\end{align}
where $\rho_\lambda$ and $D_\lambda$ are irreducible representations of $\mathfrak{S}_n$ and $GL_n(\mathbb{C})$, respectively.  For economy we henceforth write the partition $\lambda$ for the irrep 
$\rho_\lambda$ of $\mathfrak{S}_n$.
When restricted to the $(1,1,\ldots, 1)$ subspace, the $n$-fold product $\otimes^n\mathbb{C}^n$ is a carrier space
for the regular representation of $\mathfrak{S}_n$~\cite{kostant2}.
In other words, it is possible to decompose the input and output states into
linear combinations of states that transform 
irreducibly under $\mathfrak{S}_n\otimes GL_n(\mathbb{C})$.

\subsection{Block diagonalization of the coincidence rate equations}
\label{subsec:block}

Since the rate matrix carries 
the regular representation of $\mathfrak{S}_n$, there exists a linear transformation $T$ that block diagonalizes it:
\begin{align}
\operatorname{rate}(\bm\tau,s;\text{p}) =&v^\dagger(s) R(\bar{\bm\tau};\text{p})\,v(s)
\nonumber \\
 =& v^\dagger(s) T^{-1} \left(T R(\bar{\bm\tau};\text{p})) T^{-1} \right) Tv(s) \nonumber \\
 =& \left(Tv(s) \right)^\dagger \left(T R(\bar{\bm\tau};\text{p}) T^{-1} \right) 
 \left( Tv(s) \right), \label{eq:rateTvs}
\end{align}
where the species label $\text{p}$ stands for either boson ($\text{b}$) or fermion (\text{f}).

 Fraktur letters are used to shorten notation: 
 \begin{align}
 \mathfrak{v}(s)=Tv(s) \qquad \hbox{and}\qquad 
 \mathfrak{R}(\bar{\bm\tau};\text{p})\equiv TR(\bar{\bm\tau};\text{p}) 
 T^{-1}\, .   
 \end{align}
 We write $\mathfrak{R}^\lambda(\bar{\bm\tau})$ to denote the block 
 that corresponds to irrep $\lambda$ 
 which appears
 $\text{s}_\lambda$ times in the block diagonalization. 
The set of matrices $\mathfrak{R}^\lambda(\bar{\bm\tau})$ are the same for both bosons and fermions; however, their placement in the block-diagonalized rate matrix $\mathfrak{R}(\bar{\bm\tau};\text{p})$ depends on the species label $\text{p}$ and the choice of linear transformation $T$.  
 
 In the bosonic
 case, when the multiplication $\mathfrak{R}(\bar{\bm\tau};\text{b}) \mathfrak{v}(s)$ is
 preformed, the $i^\text{th}$ copy of the matrix
 $\mathfrak{R}^{\lambda}(\bar{\bm\tau})$ sees $\text{s}_\lambda$ entries of the
 vector $\mathfrak{v}(s)$; we take these entries and construct the
 vector $\mathfrak{v}_{\lambda;i}(s)$. The notation
 $\mathfrak{v}_{\lambda;i}^j(s)$ 
 refers to the $j^\text{th}$ entry of the
 vector $\mathfrak{v}_{\lambda;i}(s)$.  
 The rate is then written as 
\begin{align}
&\operatorname{rate}(\bar{\bm\tau},s;\text{b})=
 \mathfrak{v}^\dagger(s)\mathfrak{R}(\bar{\bm{\tau}};\text{b}) \mathfrak{v}(s) \\
                          =& \begin{bmatrix}
                          \phantom{\Yvcentermath1\Yboxdim{7pt}\yng(1,1,1)} \mathfrak{v}_{\Yvcentermath1\Yboxdim{3pt}\yng(3);1}(s) \\
                           \phantom{\Yvcentermath1\Yboxdim{7pt}\yng(1,1,1)} \mathfrak{v}_{\Yvcentermath1\Yboxdim{3pt}\yng(1,1,1);1}(s) \\
                           \phantom{\Yvcentermath1\Yboxdim{7pt}\yng(1,1,1)} \mathfrak{v}_{\Yvcentermath1\Yboxdim{3pt}\yng(2,1);1}^1 (s)\\
                \phantom{\Yvcentermath1\Yboxdim{7pt}\yng(1,1,1)}\mathfrak{v}_{\Yvcentermath1\Yboxdim{3pt}\yng(2,1);1}^2(s)  \\
                          \phantom{\Yvcentermath1\Yboxdim{7pt}\yng(1,1,1)} \mathfrak{v}_{\Yvcentermath1\Yboxdim{3pt}\yng(2,1);2}^1(s) \\
                            \phantom{\Yvcentermath1\Yboxdim{7pt}\yng(1,1,1)} \mathfrak{v}_{\Yvcentermath1\Yboxdim{3pt}\yng(2,1);2}^2(s) 
                           \end{bmatrix}^\dagger \begin{bmatrix}
                         \phantom{\Yvcentermath1\Yboxdim{7pt}\yng(1,1,1)} \mathfrak{R}^{\Yvcentermath1\Yboxdim{3pt}\yng(3)}
                         (\bar{\bm\tau})& 0 & 0 & 0 & 0 & 0  \\
                         \phantom{\Yvcentermath1\Yboxdim{7pt}\yng(1,1,1)} 0 & \mathfrak{R}^{\Yvcentermath1\Yboxdim{3pt}\yng(1,1,1)} (\bar{\bm\tau}) & 0 & 0 & 0 & 0 \\
                        \phantom{\Yvcentermath1\Yboxdim{7pt}\yng(1,1,1)}  0 & 0 & \mathfrak{R}^{\Yvcentermath1\Yboxdim{3pt}\yng(2,1)}_{1,1} (\bar{\bm\tau})& \mathfrak{R}^{\Yvcentermath1\Yboxdim{3pt}\yng(2,1)}_{1,2}(\bar{\bm\tau}) & 0 & 0 \\
                         \phantom{\Yvcentermath1\Yboxdim{7pt}\yng(1,1,1)} 0 & 0 & \mathfrak{R}^{\Yvcentermath1\Yboxdim{3pt}\yng(2,1)}_{2,1} (\bar{\bm\tau})& \mathfrak{R}^{\Yvcentermath1\Yboxdim{3pt}\yng(2,1)}_{2,2}(\bar{\bm\tau}) & 0 & 0 \\
                         \phantom{\Yvcentermath1\Yboxdim{7pt}\yng(1,1,1)} 0 & 0 & 0 & 0 & \mathfrak{R}^{\Yvcentermath1\Yboxdim{3pt}\yng(2,1)}_{1,1}(\bar{\bm\tau}) & \mathfrak{R}^{\Yvcentermath1\Yboxdim{3pt}\yng(2,1)}_{1,2}(\bar{\bm\tau}) \\
                         \phantom{\Yvcentermath1\Yboxdim{7pt}\yng(1,1,1)}  0 & 0 & 0 & 0 & \mathfrak{R}^{\Yvcentermath1\Yboxdim{3pt}\yng(2,1)}_{2,1}(\bar{\bm\tau}) & \mathfrak{R}^{\Yvcentermath1\Yboxdim{3pt}\yng(2,1)}_{2,2}(\bar{\bm\tau})
                           \end{bmatrix} \begin{bmatrix}
                          \phantom{\Yvcentermath1\Yboxdim{7pt}\yng(1,1,1)} \mathfrak{v}_{\Yvcentermath1\Yboxdim{3pt}\yng(3);1} (s)\\
                           \phantom{\Yvcentermath1\Yboxdim{7pt}\yng(1,1,1)} \mathfrak{v}_{\Yvcentermath1\Yboxdim{3pt}\yng(1,1,1);1}(s) \\
                           \phantom{\Yvcentermath1\Yboxdim{7pt}\yng(1,1,1)} \mathfrak{v}_{\Yvcentermath1\Yboxdim{3pt}\yng(2,1);1}^1(s) \\
                \phantom{\Yvcentermath1\Yboxdim{7pt}\yng(1,1,1)}\mathfrak{v}_{\Yvcentermath1\Yboxdim{3pt}\yng(2,1);1}^2(s)  \\
                          \phantom{\Yvcentermath1\Yboxdim{7pt}\yng(1,1,1)} \mathfrak{v}_{\Yvcentermath1\Yboxdim{3pt}\yng(2,1);2}^1(s) \\
                            \phantom{\Yvcentermath1\Yboxdim{7pt}\yng(1,1,1)} \mathfrak{v}_{\Yvcentermath1\Yboxdim{3pt}\yng(2,1);2}^2(s) 
                           \end{bmatrix} \nonumber
\end{align}

The fermion case has one key difference from the boson case.
Recall that characters of conjugate representations of $\mathfrak{S}_n$ only
differ by a sign in their odd permutations, and that the bosonic
and fermionic rate matrices have this symmetry. As a result when
the fermionic rate matrix is block-diagonalized into irreducible
representations by $T$, 
\emph{every  matrix representation  $\mathfrak{R}^{\lambda}$ that
appears in the block-diagonalization for boson is replaced by its conjugate
$\mathfrak{R}^{\lambda^*}$ for fermions}. Thus, the coincidence rate equations for the $3$-fermion case has
the blocks $\mathfrak{R}^{\Yvcentermath1\Yboxdim{3pt}\yng(1,1,1)
}(\bar{\bm\tau})$ and $\mathfrak{R}^{\Yvcentermath1\Yboxdim{3pt}\yng(3)}(\bar{\bm\tau})$
interchanged:
\begin{align}
\label{eq: fermionmatrix}
&\operatorname{rate}(\bar{\bm\tau},s;\text{f})= 
\mathfrak{v}^\dagger(s)\mathfrak{R}(\bar{\bm{\tau}};\text{f}) \mathfrak{v}(s) \\
                          =& \begin{bmatrix}
                          \phantom{\Yvcentermath1\Yboxdim{7pt}\yng(1,1,1)} \mathfrak{v}_{\Yvcentermath1\Yboxdim{3pt}\yng(3);1} (s)\\
                           \phantom{\Yvcentermath1\Yboxdim{7pt}\yng(1,1,1)} \mathfrak{v}_{\Yvcentermath1\Yboxdim{3pt}\yng(1,1,1);1}(s) \\
                           \phantom{\Yvcentermath1\Yboxdim{7pt}\yng(1,1,1)} \mathfrak{v}_{\Yvcentermath1\Yboxdim{3pt}\yng(2,1);1}^1(s) \\
                \phantom{\Yvcentermath1\Yboxdim{7pt}\yng(1,1,1)}
                \mathfrak{v}_{\Yvcentermath1\Yboxdim{3pt}\yng(2,1);1}^2 (s) \\
                \phantom{\Yvcentermath1\Yboxdim{7pt}\yng(1,1,1)}
                \mathfrak{v}_{\Yvcentermath1\Yboxdim{3pt}\yng(2,1);2}^1 (s)\\
                \phantom{\Yvcentermath1\Yboxdim{7pt}\yng(1,1,1)} \mathfrak{v}_{\Yvcentermath1\Yboxdim{3pt}\yng(2,1);2}^2 (s)
                           \end{bmatrix}^\dagger 
        \begin{bmatrix}
        \phantom{\Yvcentermath1\Yboxdim{7pt}\yng(1,1,1)}
        \mathfrak{R}^{\Yvcentermath1\Yboxdim{3pt}\yng(1,1,1)}(\bar{\bm{\tau}}) & 0 & 0 & 0 & 0 & 0  \\
        \phantom{\Yvcentermath1\Yboxdim{7pt}\yng(1,1,1)} 0 & \mathfrak{R}^{\Yvcentermath1\Yboxdim{3pt}\yng(3)}(\bar{\bm{\tau}}) & 0 & 0 & 0 & 0 \\
        \phantom{\Yvcentermath1\Yboxdim{7pt}\yng(1,1,1)}  0 & 0 & \mathfrak{R}^{\Yvcentermath1\Yboxdim{3pt}\yng(2,1)}_{1,1}(\bar{\bm{\tau}}) & \mathfrak{R}^{\Yvcentermath1\Yboxdim{3pt}\yng(2,1)}_{1,2}(\bar{\bm{\tau}}) & 0 & 0 \\
        \phantom{\Yvcentermath1\Yboxdim{7pt}\yng(1,1,1)}  0 & 0 & \mathfrak{R}^{\Yvcentermath1\Yboxdim{3pt}\yng(2,1)}_{2,1}(\bar{\bm{\tau}}) & \mathfrak{R}^{\Yvcentermath1\Yboxdim{3pt}\yng(2,1)}_{2,2}(\bar{\bm{\tau}}) & 0 & 0 \\
        \phantom{\Yvcentermath1\Yboxdim{7pt}\yng(1,1,1)}  0 & 0 & 0 & 0 & \mathfrak{R}^{\Yvcentermath1\Yboxdim{3pt}\yng(2,1)}_{1,1} (\bar{\bm{\tau}})& \mathfrak{R}^{\Yvcentermath1\Yboxdim{3pt}\yng(2,1)}_{1,2}(\bar{\bm{\tau}}) \\
        \phantom{\Yvcentermath1\Yboxdim{7pt}\yng(1,1,1)}  0 & 0 & 0 & 0 & \mathfrak{R}^{\Yvcentermath1\Yboxdim{3pt}\yng(2,1)}_{2,1}(\bar{\bm{\tau}}) & \mathfrak{R}^{\Yvcentermath1\Yboxdim{3pt}\yng(2,1)}_{2,2}(\bar{\bm{\tau}})
        \end{bmatrix} 
        \begin{bmatrix}
        \phantom{\Yvcentermath1\Yboxdim{7pt}\yng(1,1,1)} \mathfrak{v}_{\Yvcentermath1\Yboxdim{3pt}\yng(3);1}(s) \\
        \phantom{\Yvcentermath1\Yboxdim{7pt}\yng(1,1,1)} \mathfrak{v}_{\Yvcentermath1\Yboxdim{3pt}\yng(1,1,1);1}(s) \\
        \phantom{\Yvcentermath1\Yboxdim{7pt}\yng(1,1,1)} \mathfrak{v}_{\Yvcentermath1\Yboxdim{3pt}\yng(2,1);1}^1(s) \\
        \phantom{\Yvcentermath1\Yboxdim{7pt}\yng(1,1,1)}
        \mathfrak{v}_{\Yvcentermath1\Yboxdim{3pt}\yng(2,1);1}^2(s)  \\
        \phantom{\Yvcentermath1\Yboxdim{7pt}\yng(1,1,1)} \mathfrak{v}_{\Yvcentermath1\Yboxdim{3pt}\yng(2,1);2}^1(s) \\
        \phantom{\Yvcentermath1\Yboxdim{7pt}\yng(1,1,1)} \mathfrak{v}_{\Yvcentermath1\Yboxdim{3pt}\yng(2,1);2}^2(s) 
        \end{bmatrix}. \nonumber
\end{align}

When each $s_i=1$ or $0$, the $\mathfrak{v}_{\lambda,i}(s)$ are
group functions of weight $(1,1,1,\ldots, 1)$ for the the irrep $\lambda$.  Thus
for any value of $n$, the respective coincidence rate equations for bosons and fermions take the general form: 
\begin{align}
&\operatorname{rate}(\bar{\bm\tau},s;\text{b})= \sum_{\lambda \vdash n} \sum_{i=1}^{\text{s}_\lambda} \mathfrak{v}_{\lambda;i}^\dagger(s)\, 
\mathfrak{R}^\lambda(\bar{\bm{\tau}})\,  \mathfrak{v}_{\lambda;i}(s),\label{eq:rbn} \\
&\qquad = \vert\text{per}(A(s))\vert^2\text{per}(r(\bar{\bm\tau}))+ [\hbox{other $\mathfrak{S}_n$ irreps}]+  \vert\hbox{det}(A(s))\vert^2 \hbox{det}(r(\bar{\bm\tau})) \, , \nonumber \\
&\operatorname{rate}(\bar{\bm\tau},s;\text{f})= \sum_{\lambda \vdash n} \sum_{i=1}^{\text{s}_\lambda} \mathfrak{v}_{\lambda;i}^\dagger(s)\, \mathfrak{R}^{\lambda^*}(\bar{\bm\tau})\,  \mathfrak{v}_{\lambda;i}(s), \label{eq:rfn} \\
& \qquad = \vert\hbox{det}(A(s))\vert^2\text{per}(r(\bar{\bm\tau}))
+ [\hbox{other $\mathfrak{S}_n$ irreps}] + \vert\text{per}(A(s))\vert^2 \hbox{det}(r(\bar{\bm\tau}))   \,. \nonumber
\end{align}

The notation $\lambda \vdash n$ means that we sum over all partitions of $n$, and $\lambda^*$ is the conjugate partition of $\lambda$. In particular, when all relative arrival times are identical, \emph{i.e.} when all particles are exactly indistinguishable, $\text{per}(r(\bar{\bm\tau}))=n!$ and 
all other terms in the sums of Eqs.~(\ref{eq:rbn}) and~(\ref{eq:rfn}) are $0$;
in the case of bosons,
Eq.~(\ref{eq:rbn}) is then truncated to the modulus square of the permanent 
of $A(s)$, and we recover the original
\textsc{BosonSampling} result; in the fermion case, only
$\vert\hbox{det}(A(s))\vert^2$ survives.

\subsection{Immanants and $\mathcal{D}$-functions for $\mathfrak{R}^\lambda(\bar{\bm\tau})$ 
and $\mathfrak{v}_\lambda$}
\label{subsec:immanants}

In this subsection we again assume $s$ so that $s_i=0$ or $1$. 
The result holds for any such $s$ and the string label $s$ is implicit
throughout. We show that every entry of the matrix $\mathfrak{R}^\lambda(\bar{\bm\tau})$ is a linear
combination of permuted $\lambda$-immanants of the delay matrix $r(\bar{\bm\tau})$; and
each $\mathfrak{v}_{\lambda;i}$ is a vector where each entry is a linear
combination of permuted $\lambda$-immanants of the scattering matrix $A\in GL_n(\mathbb{C})$. 
The elements of the matrix $\mathfrak{R}^\lambda(\bar{\bm\tau})$ and the vector $\mathfrak{v}$ are
also $GL_n(\mathbb{C})$ group functions, also referred to in the physics
literature as Wigner $\mathcal{D}$-functions~\cite{chacon1966representations,rowe1999representations,wigner1959group}. 

Let $\{\vert \psi_i^\lambda \rangle, i=1,\ldots,\text{s}_\lambda\}$ denote the set of $(1,1,\ldots,1)$ basis states in the $GL_n(\mathbb{C})$ irrep $\lambda$.  For $g\in GL_n(\mathbb{C})$ the overlap of Eq.~(\ref{eq: innerprod}) specializes to
\begin{align}
{\cal D}^\lambda_{ij}(g):=\langle \psi_i^\lambda \vert D_\lambda(g) \vert \psi_j^\lambda \rangle
\label{eq:D111def}
\end{align}
where $D_\lambda(g)$ is the representation of $g$ by the matrix $D_\lambda$ in the irrep $\lambda$. 

\begin{example}
For an arbitrary matrix $Z\in GL_3(\mathbb{C})$, the
$\mathcal{D}_{ij}$ functions are given in Table \ref{table:Ds} where:
\begin{align}
\biggl\vert \begin{array}{ccccc}
    3& &0& &0\\
     &2& &0&\\
     &&1&&\end{array}\biggr\rangle &\to \vert \psi^{\Yvcentermath1\Yboxdim{3pt}\yng(3)}_1\rangle, \\
    \biggl\vert \begin{array}{ccccc}
    2& &1& &0\\
     &2& &0&\\
     &&1&&\end{array}\biggr\rangle &\to \vert \psi^{\Yvcentermath1\Yboxdim{3pt}\yng(2,1)}_1\rangle\qquad
    \biggl\vert \begin{array}{ccccc}
    2& &1& &0\\
     &1& &1&\\
     &&1&&\end{array}\biggr\rangle \to \vert \psi^{\Yvcentermath1\Yboxdim{3pt}\yng(2,1)}_2\rangle,\\
     \biggl\vert \begin{array}{ccccc}
    1& &1& &1\\
     &1& &1&\\
     &&1&&\end{array}\biggr\rangle &\to \vert \psi^{\Yvcentermath1\Yboxdim{3pt}\yng(1,1,1)}_1\rangle. 
\end{align}

\noindent are the Gelfand-Tsetlin patterns~\cite{alex2011numerical,barut1986theory,gelfand1950matrix} 
{for the} $(1,1,1)$-states in the $\Yvcentermath1\Yboxdim{4pt}\yng(3), \yng(2,1),$ and $\Yvcentermath1\Yboxdim{4pt} \yng(1,1,1)$ {subspaces}, respectively.

\begin{table}[h!]
    \centering
    {\renewcommand{\arraystretch}{1.95}
        \begin{tabular}{|c|c|}
        \hline
        $\mathcal{D}$-function Notation & Polynomial in $Z$ entries \\
           \hline
             ${\cal D}^{\Yvcentermath1\Yboxdim{3pt}\yng(1,1,1)}_{1,1}(Z)$ & $\begin{array}{l}
             Z_{11} Z_{22} Z_{33} - Z_{12} Z_{21} Z_{33} - Z_{13} Z_{22} Z_{31}  \\ \qquad -Z_{11}Z_{23}Z_{32} + Z_{12}Z_{23}Z_{31} + Z_{13}Z_{21}Z_{32}\end{array}$ \\
            \hline
             ${\cal D}^{\Yvcentermath1\Yboxdim{3pt}\yng(3)}_{1,1}(Z)$ & 
             $\begin{array}{l}Z_{11} Z_{22} Z_{33} + Z_{12} Z_{21} Z_{33} + Z_{13} Z_{22} Z_{31}  \\ \qquad +Z_{11}Z_{23}Z_{32} + Z_{12}Z_{23}Z_{31} + Z_{13}Z_{21}Z_{32}\end{array}$ \\
            \hline
            ${\cal D}^{\Yvcentermath1\Yboxdim{3pt}\yng(2,1)}_{1,1}(Z)$ &
            $\begin{array}{l}\frac{1}{2}\left(-Z_{13} Z_{22} Z_{31}-Z_{12} Z_{23} Z_{31}-Z_{13} Z_{21} Z_{32}\right.\\ \qquad \left. -Z_{11} Z_{23} Z_{32}+2 Z_{12}
   Z_{21} Z_{33}+2 Z_{11} Z_{22} Z_{33}\right)\end{array}$ \\
   \hline
            ${\cal D}^{\Yvcentermath1\Yboxdim{3pt}\yng(2,1)}_{1,2}(Z)$ &$\begin{array}{l}\frac{\sqrt{3}}{2}
            \left (-Z_{13} Z_{22} Z_{31}-Z_{12} Z_{23} Z_{31} \right.\\
            \qquad\qquad  \left. +Z_{13} Z_{21} Z_{32}+Z_{11} Z_{23}Z_{32}\right)\end{array}$\\ 
            \hline 
   ${\cal D}^{\Yvcentermath1\Yboxdim{3pt}\yng(2,1)}_{2,1}(Z)$ & $\begin{array}{l}\frac{\sqrt{3}}{2}(-Z_{13} Z_{22} Z_{31}+Z_{12} Z_{23} Z_{31} \\ 
   \qquad \qquad -Z_{13} Z_{21} Z_{32}+Z_{11} Z_{23}Z_{32} )\end{array}$ \\
   \hline
   ${\cal D}^{\Yvcentermath1\Yboxdim{3pt}\yng(2,1)}_{2,2}(Z)$ & $
   \begin{array}{l}\frac{1}{2}(Z_{13} Z_{22} Z_{31}-Z_{12} Z_{23} Z_{31}-Z_{13} Z_{21} Z_{32} 
   \\ \qquad Z_{11} Z_{23} Z_{32}-2 Z_{12}
   Z_{21} Z_{33}+2 Z_{11} Z_{22} Z_{33})\end{array}$ \\
   \hline
        \end{tabular}
        }
        \caption{Group functions ${\cal D}^\lambda _{ij}(Z)$ of Eq.~(\ref{eq:D111def}) connecting $(1,1,1)$- states
        in $GL_3(\mathbb{C})$.}
     \label{table:Ds}
    \end{table}
\end{example}

We use Theorem \ref{theorem:Kostant} to prove two results about $\mathcal{D}$-functions. Firstly: 
\begin{lemma}
\label{lemma:everyelement}
Permuted immanants are linear combinations of $\mathcal{D}$-functions (and vice versa). 
\end{lemma}
\begin{proof}
Let $P_{\sigma}$ be the permutation matrix corresponding to the permutation $\sigma$. Depending on whether $P_{\sigma}$ acts on the left or the right of a matrix 
$M$, either the rows or columns of $M$ get permuted. 
In this proof, $P_{\sigma}$ acts on the left, which permutes the rows. To prove this lemma,
Theorem \ref{theorem:Kostant} is applied to a row-permuted matrix $P_{\sigma} M$: 
\begin{align}
\operatorname{imm}^\lambda\left(P_{\sigma} M \right) =& \sum_{i} \mathcal{D}_{i,i}^\lambda(P_{\sigma} M) \\
   =& \sum_i \bra{\psi_i^\lambda} D_\lambda(P_{\sigma})  D_\lambda(M)  \ket{\psi_i^\lambda} 
   =\sum_{ij}{\cal D}^\lambda_{i,j}(P_\sigma){\cal D}_{j,i}^\lambda (M)\, . 
\end{align}
Note that the coefficients ${\cal D}^\lambda_{ij}(P_\sigma) =
\bra{\psi_i^\lambda} D_\lambda(P_{\sigma}) \ket{\psi_j^\lambda}$ are the
standard (Yamanouchi) entries \cite{chen1989group} for the matrix
representation of $P_\sigma$ in the representation $\lambda$ of
$\mathfrak{S}_n$ . From the above computation, we see that an immanant of
$P_{\sigma} M$ is a linear combination of $\mathcal{D}$-functions, which in
turn means that all the entries of $\mathcal{D}^\lambda(M)$ are linear
combinations of permuted $\lambda$-immanants. In particular, for an
arbitrary matrix, this corollary shows that of the total $n!$ possible
permuted immanants, there are exactly $\text{s}_\lambda^2$ linearly independent
permuted immanants of shape $\lambda$, as there are that many
$\mathcal{D}$-functions.
\end{proof}

We give an example for the $3$-fermion case, where the above lemma implies
that every element of $\mathfrak{R}^{\lambda}(\bar{\bm\tau})$ is a linear combination of
permuted immanants of the delay matrix $r(\bar{\bm\tau})$. To simplify notation,
$\lambda_{\sigma}^M$ is used to denote the $\lambda$-immanant of the matrix
$M$ whose rows are permuted by $\sigma$, that is, row $i$ is replaced by row $\sigma(i)$. There are $3!=6$ permuted
immanants, but as we have shown, only a linear combination of $\text{s}_\lambda^2$
of them are needed to write an entry of $\mathfrak{R}^\lambda(\bar{\bm\tau})$; we choose
the permutations $e$, $(12)$, $(23)$, and $(132)$. The entries of
$\mathfrak{R}(\bar{\bm\tau};f)$ that appear in Eq.~(\ref{eq: fermionmatrix}) as both a
linear combination of permuted immanants and as a $\mathcal{D}$-function
are given in Table \ref{fig:immanantsandDfunctions}.

 \begin{table}[h!]
    \centering
    {\renewcommand{\arraystretch}{1.85}
        \begin{tabular}{|c|c|c|c|}
        \hline
           \shortstack{Entry\\ in $\mathfrak{R}(\bar{\bm\tau};\text{p})$}   & 
           \shortstack{$\mathcal{D}$-function \\ Notation}
            & Sum of Permuted Immanants & Polynomial in $r(\bar{\bm\tau})$ entries \\
           \hline
            $\mathfrak{R}^{\Yvcentermath1\Yboxdim{3pt}\yng(1,1,1)}(\bar{\bm\tau})$ & ${\cal D}^{\Yvcentermath1\Yboxdim{3pt}\yng(1,1,1)}_{1,1}(r)$ & $\hbox{det}(r)$& 
            $\begin{array}{l} 1 - r_{12}^2 - r_{13}^2  \\ \quad -r_{23}^2 + 2r_{12}r_{23}r_{31}\end{array}$ \\
            \hline
            $\mathfrak{R}^{\Yvcentermath1\Yboxdim{3pt}\yng(3)}(\bar{\bm\tau})$  & ${\cal D}^{\Yvcentermath1\Yboxdim{3pt}\yng(3)}_{1,1}
(r)$ & $\text{per}(r)$ & 
$\begin{array}{l} 1 + r_{12}^2 + r_{13}^2  \\ \quad +r_{23}^2 + 2r_{12}r_{23}r_{31}\end{array}$ \\
            \hline
            $\mathfrak{R}^{\Yvcentermath1\Yboxdim{3pt}\yng(2,1)}_{1,1}
            (\bar{\bm\tau})$  & ${\cal D}^{\Yvcentermath1\Yboxdim{3pt}\yng(2,1)}_{1,1}
(r)$ & $\frac{1}{2}\left( \Yvcentermath1\Yboxdim{4pt}\yng(2,1)_{e}^{\,r} + \Yvcentermath1\Yboxdim{4pt}\yng(2,1)_{(12)}^{\,r} \right)$ & 
$\begin{array}{ll} 1 + r_{12}^2  - \frac{1}{2}r_{13}^2 \\ 
\qquad - \frac{1}{2}r_{23}^2 - r_{12}r_{23}r_{31}\end{array}$ \\
            \hline
            $\mathfrak{R}^{\Yvcentermath1\Yboxdim{3pt}\yng(2,1)}_{1,2}
            (\bar{\bm\tau})$  
            & ${\cal D}^{\Yvcentermath1\Yboxdim{3pt}\yng(2,1)}_{1,2}(r)$ & 
$\begin{array}{l}\sqrt{3}\left( -\frac{1}{6} \Yvcentermath1\Yboxdim{4pt}\yng(2,1)_{e}^{\,r} +\frac{1}{3} \Yvcentermath1\Yboxdim{4pt}\yng(2,1)_{(23)}^{\,r} \right. \\
\left. \qquad +\frac{1}{6} \Yvcentermath1\Yboxdim{4pt}\yng(2,1)_{(12)}^{\,r} - \frac{1}{3} \Yvcentermath1\Yboxdim{4pt}\yng(2,1)_{(132)}^{\,r} \right)\end{array}$ & $\frac{\sqrt{3}}{2}\left(r_{23}^2 - r_{13}^2\right)$\\
            \hline
            $\mathfrak{R}^{\Yvcentermath1\Yboxdim{3pt}\yng(2,1)}_{2,1}
            (\bar{\bm\tau})$  & ${\cal D}^{\Yvcentermath1\Yboxdim{3pt}\yng(2,1)}_{2,1}(r)$ 
            & 
$\begin{array}{l}
\sqrt{3}\left( -\frac{1}{6} \Yvcentermath1\Yboxdim{4pt}\yng(2,1)_{e}^{\,r} +\frac{1}{3} \Yvcentermath1\Yboxdim{4pt}\yng(2,1)_{(23)}^{\,r} \right. \\
\qquad \left. +\frac{1}{6} \Yvcentermath1\Yboxdim{4pt}\yng(2,1)_{(12)}^{\,r} - \frac{1}{3} \Yvcentermath1\Yboxdim{4pt}\yng(2,1)_{(132)}^{\,r} \right)\end{array}$ & $\frac{\sqrt{3}}{2}\left(r_{23}^2 - r_{13}^2\right)$ \\
            \hline
            $\mathfrak{R}^{\Yvcentermath1\Yboxdim{3pt}\yng(2,1)}_{2,2} (\bar{\bm\tau}) $  & ${\cal D}^{\Yvcentermath1\Yboxdim{3pt}\yng(2,1)}_{2,2}
(r)$ & $ \frac{1}{2}\left( \Yvcentermath1\Yboxdim{4pt}\yng(2,1)_{e}^{\,r} - \Yvcentermath1\Yboxdim{4pt}\yng(2,1)_{(12)}^{\,r} \right)$ & 
$\begin{array}{l} 1 - r_{12}^2  + \frac{1}{2}r_{13}^2 \\ \qquad + \frac{1}{2}r_{23}^2 - r_{12}r_{23}r_{31}\end{array}$ \\
            \hline
        \end{tabular}}
        \caption{Explicit expressions of various 
        $\mathfrak{R}^{\lambda}_{ij}(\bar{\bm\tau})$ for $n=3$ particles. The argument $\bar{\bm\tau}$ is implicit
        in $r$ or $r_{ij}$.}
     \label{fig:immanantsandDfunctions}
    \end{table}

In Table \ref{fig:immanantsandDfunctions}, the equality $r_{ij}(\bar{\bm\tau})=r_{ji}(\bar{\bm\tau})$ has
been used to simplify the polynomials. Notice how $\mathfrak{R}^{\Yvcentermath1\Yboxdim{3pt}\yng(2,1)}_{1,2}(\bar{\bm\tau})
=\mathfrak{R}^{\Yvcentermath1\Yboxdim{3pt}\yng(2,1)}_{2,1}(\bar{\bm\tau})$ and that there are $5$
distinct functions that appear in the in the block-diagonalized rate matrix. In the general case
the rate matrix is a real symmetric matrix, so its block-diagonal form is also be symmetric. To
determine the number of distinct functions we simply need to count the number of entries on or
above the main diagonal in $\mathfrak{R}^{\lambda}(\bar{\bm\tau})$ for each $\lambda \vdash n$. Each
$\mathfrak{R}^{\lambda}(\bar{\bm\tau})$ matrix is of size $\text{s}_\lambda$, 
so by summing over all $\lambda
\vdash n$ we find that the number of distinct functions in the block-diagonalized rate matrix is 
\begin{align}
\label{eq:exactnumber}
\sum_{\lambda \vdash n} \frac{\text{s}_\lambda^2 + \text{s}_\lambda}{2}\sim \frac{n!}{2},
\end{align}
as $\sum_{\lambda \vdash n} \text{s}_\lambda^2=n!$.

The same linear transformation $T$ that block-diagonalizes the rate matrix
also acts on the scattering matrix $A(s)$. Similarly, we see that
$\mathcal{D}^\lambda(A(s))$ is the matrix whose $i^\text{th}$ column is
$\mathfrak{v}_{\lambda;i}(s)$. For the column vector $\mathfrak{v}(s)$, one also
easily expresses the entries in terms of permuted immanants or ${\cal
D}$-functions:
\begin{align}
\begin{bmatrix}
                          \phantom{\Yvcentermath1\Yboxdim{7pt}\yng(1,1,1)} \mathfrak{v}_{\Yvcentermath1\Yboxdim{3pt}\yng(3);1} (s)\\
                           \phantom{\Yvcentermath1\Yboxdim{7pt}\yng(1,1,1)} \mathfrak{v}_{\Yvcentermath1\Yboxdim{3pt}\yng(1,1,1);1}(s) \\
                           \phantom{\Yvcentermath1\Yboxdim{7pt}\yng(1,1,1)} \mathfrak{v}_{\Yvcentermath1\Yboxdim{3pt}\yng(2,1);1}^1(s) \\
                \phantom{\Yvcentermath1\Yboxdim{7pt}\yng(1,1,1)}
                \mathfrak{v}_{\Yvcentermath1\Yboxdim{3pt}\yng(2,1);1}^2 (s) \\
                \phantom{\Yvcentermath1\Yboxdim{7pt}\yng(1,1,1)}
                \mathfrak{v}_{\Yvcentermath1\Yboxdim{3pt}\yng(2,1);2}^1 (s)\\
                \phantom{\Yvcentermath1\Yboxdim{7pt}\yng(1,1,1)} \mathfrak{v}_{\Yvcentermath1\Yboxdim{3pt}\yng(2,1);2}^2 (s)
                           \end{bmatrix} = \begin{bmatrix}
\phantom{\Yvcentermath1\Yboxdim{4pt}\yng(1,1,1,1,1,1)} \frac{1}{\sqrt{6}}\Yvcentermath1\Yboxdim{4pt}\yng(3)^{A(s)} \\
\phantom{\Yvcentermath1\Yboxdim{4pt}\yng(1,1,1,1,1,1)} \frac{1}{\sqrt{6}} \Yvcentermath1\Yboxdim{4pt}\yng(1,1,1)^{A(s)} \\
\phantom{\Yvcentermath1\Yboxdim{4pt}\yng(1,1,1,1,1,1)}\frac{1}{2 \sqrt{3}} \left(\Yvcentermath1\Yboxdim{4pt}\yng(2,1)_e^{A(s)} + \Yvcentermath1\Yboxdim{4pt}\yng(2,1)_{(12)}^{A(s)}  \right) \\
\phantom{\Yvcentermath1\Yboxdim{4pt}\yng(1,1,1,1,1,1)}-\frac{1}{6} \Yvcentermath1\Yboxdim{4pt}\yng(2,1)_{e}^{A(s)} +\frac{1}{3} \Yvcentermath1\Yboxdim{4pt}\yng(2,1)_{(23)}^{A(s)} +\frac{1}{6} \Yvcentermath1\Yboxdim{4pt}\yng(2,1)_{(12)}^{A(s)} - \frac{1}{3} \Yvcentermath1\Yboxdim{4pt}\yng(2,1)_{(132)}^{A(s)} \\
\phantom{\Yvcentermath1\Yboxdim{4pt}\yng(1,1,1,1,1,1)}\frac{1}{6} \Yvcentermath1\Yboxdim{4pt}\yng(2,1)_{e}^{A(s)} + \frac{1}{3} \Yvcentermath1\Yboxdim{4pt}\yng(2,1)_{(23)}^{A(s)} + \frac{1}{6} \Yvcentermath1\Yboxdim{4pt}\yng(2,1)_{(12)}^{A(s)} + \frac{1}{3} \Yvcentermath1\Yboxdim{4pt}\yng(2,1)_{(132)}^{A(s)} \\
\phantom{\Yvcentermath1\Yboxdim{4pt}\yng(1,1,1,1,1,1)}\frac{1}{2 \sqrt{3}} \left( \Yvcentermath1\Yboxdim{4pt}\yng(2,1)_e^{A(s)} - \Yvcentermath1\Yboxdim{4pt}\yng(2,1)_{(12)}^{A(s)}  \right)
\end{bmatrix} 
= 
{\renewcommand{\arraystretch}{1.85}\begin{bmatrix}
\frac{1}{\sqrt{6}}{\cal D}^{\Yvcentermath1\Yboxdim{4pt}\yng(3)}(A(s)) \\
\frac{1}{\sqrt{6}}{\cal D}^{\Yvcentermath1\Yboxdim{4pt}\yng(1,1,1)}(A(s)) \\
\frac{1}{\sqrt{3}}{\cal D}_{1,1}^{\Yvcentermath1\Yboxdim{4pt}\yng(2,1)}(A(s)) \\
\frac{1}{\sqrt{3}}{\cal D}_{1,2}^{\Yvcentermath1\Yboxdim{4pt}\yng(2,1)}(A(s)) \\
\frac{1}{\sqrt{3}}{\cal D}_{2,1}^{\Yvcentermath1\Yboxdim{4pt}\yng(2,1)}(A(s)) \\
\frac{1}{\sqrt{3}}{\cal D}_{2,2}^{\Yvcentermath1\Yboxdim{4pt}\yng(2,1)}(A(s))
\end{bmatrix}
}
\end{align}

In general, the entries of $\mathfrak{v}_{\lambda;i}(s)$ are group functions multiplied by the scaling factor $\sqrt{\text{s}_\lambda/n!}$, where 
$\text{s}_\lambda=\hbox{dim}(\lambda)$ is the dimension 
of the irrep $\lambda$ of $\mathfrak{S}_n$.

\begin{lemma}
\label{lemma:zeroentries}
Assume $s$ with $s_i=0$ or $1$, and
let $r(\bar{\bm\tau})$ be the $n \times n$ delay matrix and $\mathfrak{R}^\lambda(\bar{\bm\tau})$ be a 
block that
carries the $\lambda$-representation of $\mathfrak{S}_n$ in the block diagonalization of the rate
matrix. If $\operatorname{imm}^\lambda(r(\bar{\bm\tau}))=0$, then every entry of the matrix
$\mathfrak{R}^\lambda(\bar{\bm\tau})$ is also $0$.
\end{lemma}
\begin{proof}
Begin by assuming that $\operatorname{imm}^\lambda(r(\bar{\bm\tau}))=0$. Applying
Theorem~\ref{theorem:Kostant} to the delay
matrix, it follows that the trace of $\mathfrak{R}^\lambda(\bar{\bm\tau})$ is $0$. 
From Eqs.~(\ref{eq:ratematrix}) and~(\ref{eq:ratematrixentry}), the rate matrix 
$R(\bar{\bm\tau};p)$ is clearly a Gram matrix, and remains as such under
the change of basis generated by 
the similarity transformation $T$ that brings it to block diagonal form.
Thus
$\mathfrak{R}^\lambda(\bar{\bm\tau})$ is a Gram matrix and has non-negative eigenvalues. Since the trace of a
matrix is the sum of its eigenvalues, the trace of $\mathfrak{R}^\lambda(\bar{\bm\tau})$ is $0$, and
$\mathfrak{R}^\lambda(\bar{\bm\tau})$ has all non-negative eigenvalues, it follows that all the eigenvalues of
$\mathfrak{R}^\lambda(\bar{\bm\tau})$ are zero. Since all the eigenvalues of 
$\mathfrak{R}^\lambda(\bar{\bm\tau})$ are zero, it
must be nilpotent; however, $\mathfrak{R}^\lambda(\bar{\bm\tau})$ is also symmetric and it's well-known that the only matrix that is both symmetric and nilpotent is the zero matrix. 
\end{proof}\\
{Lemma~\ref{lemma:zeroentries} thus establishes a simple test to determine which 
$\mathfrak{R}^\lambda(\boldsymbol{\bar\tau})$ are $0$, thus truncating the sum in Eqs.~(\ref{eq:rbn}) or (\ref{eq:rfn}).}

\subsection{Fully distinguishable particles}\label{subsec:distinguishable}

{There remains an important case to examine:} the situation where some particles are fully distinguishable.  This occurs when $r_{ij}(\bar{\boldsymbol{\tau}})=r_{ji}(\bar{\boldsymbol{\tau}})=0$ in the matrix $r(\bar{\boldsymbol{\tau}})$. {In this section we first show how this limit
can modify our block diagonalization algorithm to provide significant simplifications, and discuss next if this limit is realizable within our time-overlap model to compute
$r_{ij}$ for $i\ne j$.}

{We first investigate the modification to our scheme by supposing}
one particle is fully distinguishable from all the rest.  For definiteness, take this to the particle number $k$.  Then the entries $r_{ik}(\bar{\boldsymbol{\tau}})$ and $r_{ki}(\bar{\boldsymbol{\tau}})$ of the matrix $r(\bar{\boldsymbol{\tau}})$ will be $0$ except for $r_{kk}(\bar{\boldsymbol{\tau}})=1$; {the matrix $r(\boldsymbol{\tau})$ is thus block diagonal}.  The matrix $R(\bar{\boldsymbol{\tau}})$ will {also be block diagonal so that} after suitable permutations of rows and columns, it can be brought to a form with $n$ identical blocks {explicitly} repeated.   Our previous procedure then amounts to 
{further} block diagonalization inside 
{each of the repeated blocks} using representations of $\mathfrak{S}_{n-1}$ rather than $\mathfrak{S}_n$.

If two particles are fully distinguishable - say particle number $k$ and particle number $q$ - one simply applies the same observation as previously to particle number $q$,
thereby reducing the problem to multiple copies of the $\mathfrak{S}_{n-2}$ problem. Obviously, as the number of fully distinguishable particles increases, the size of each non-trivial block in the matrix $R$ decreases, until eventually one reaches the case where all particles are fully distinguishable: this is {a} ``classical'' limit. In this case, the entries $r_{ij}(\bar{\boldsymbol{\tau}})$ of the matrix $r(\bar{\boldsymbol{\tau}})$ are $\delta_{ij}$, and the rate matrix $R(\bar{\boldsymbol{\tau}})$ is already diagonal: there is thus no need to block diagonalize $R(\bar{\boldsymbol{\tau}})$. 

For instance, if we have $n=3$ fermions and the last particle is fully distinguishable, then
\begin{align}
    r(\bar{\boldsymbol{\tau}})=
    \left(\begin{array}{ccc}
    1&r_{12}(\bar{\boldsymbol{\tau}})&0\\
    r_{21}(\bar{\boldsymbol{\tau}})&1&0\\
    0&0&1\end{array}\right)\, ,
\end{align}
and, after {suitable rearrangement of} rows and columns, we reach
\begin{align}
    R(\bar{\boldsymbol{\tau}};\text{f})=
    \left(\begin{array}{cccccc}
    1& -r_{12}^2(\bar{\boldsymbol{\tau}})&0&0&0&0\\
     -r_{12}^2(\bar{\boldsymbol{\tau}})&1&0&0&0&0\\
     0&0&1& -r_{12}^2(\bar{\boldsymbol{\tau}})&0&0\\
     0&0&-r_{12}^2(\bar{\boldsymbol{\tau}})&1&0&0\\
     0&0&0&0&1& -r_{12}^2(\bar{\boldsymbol{\tau}})\\
     0&0&0&0& -r_{12}^2(\bar{\boldsymbol{\tau}})&1
    \end{array}\right)\, , \label{eq:3distinguishable}
\end{align}
where $r_{21}(\bar{\boldsymbol{\tau}})=r^*_{12}(\bar{\boldsymbol{\tau}})$ has been used.  Eq.~(\ref{eq:3distinguishable}) is clearly $3$ copies of the problem
with $2$ partially distinguishable fermions.

If now all particles are fully distinguishable, $r_{ij}(\bar{\boldsymbol{\tau}})
=\delta_{ij}$ so that $r(\bar{\boldsymbol{\tau}})$ and $R(\bar{\boldsymbol{\tau}};
\text{p})$ are both unit matrices: {there is no} need for block diagonalization. In this case the rate collapses to 
\begin{align}
    \sum_{i} \vert v_i(s)\vert^2
    =\sum_{\gamma\in\mathfrak{S}_3} \vert \Asymbol_{(\gamma)}(s)\vert^2=
    \sum_{\gamma\in \mathfrak{S}_3}\vert A_{1\gamma(1)}(s)\vert^2 \vert A_{2\gamma(2)}(s)\vert^2 \vert A_{3\gamma(3)}(s)\vert^2
\end{align}
which can be seen to be the permanent of the non-negative $n\times n$ matrix with entries $\vert A_{ij}(s)\vert^2$.  In this case, one can use an efficient algorithm \cite{PermAlg} to evaluate the permanent of this type of matrix as long as a small amount of error is allowed.

{We now ask if it is possible to obtain $r_{ij}=\delta_{ij}$ for all particles in our model when computing $r_{ij}$ as the temporal overlap 
$\text{e}^{-\delta\omega^2\,(\tau_i-\tau_j)^2/2}$.} {For simplicity, we imagine each pulse separated from the one ahead of it by 
some fixed interval $\Delta T=\tau_{i+1}-\tau_i$, independent of $i$.  Thus, under this assumption, we have $(\tau_i-\tau_j)^2= (i-j)^2\Delta T^2$
We note that, for $n$  particles, the total integration time is ${\cal T}\sim n \Delta T$,
so using ${\cal T}\Delta \Omega\approx 2\pi$ for Fourier-limited detectors yields
\begin{align}
  (\tau_i-\tau_j)^2&=\frac{(i-j)^2{\cal T}^2}{n^2}
  \sim  \frac{4\pi ^2 (i-j)^2 }{n^2 (\Delta\Omega)^2}\, ,\\
  \delta\omega^2 (\tau_i-\tau_j)^2 &\sim  \frac{\Delta\omega^2}{\Delta\omega^2+\Delta\Omega^2}
    \frac{4\pi ^2 (i-j)^2 }{n^2 }\, . \label{eq:exponentterm}
\end{align}

It will be convenient to use $1/\Delta\omega$, the inverse frequency width of the pulse, as the timescale for our problem.    
For long integration times ${\cal T}\sim 2\pi/\Delta \Omega\gg 1/\Delta \omega$, or $\Delta\omega\gg \Delta \Omega/2\pi$, and Eq.(\ref{eq:exponentterm}) becomes 
\begin{align}
   \delta\omega^2 (\tau_i-\tau_j)^2 &\approx
    \frac{4\pi ^2 (i-j)^2 }{n^2 }\, , \qquad (\hbox{for } {\cal T}\gg 1/\Delta\omega)\, ,
\end{align}
which is appreciable only if $\vert i-j\vert /n\approx 1$, \emph{i.e.} for ``very distant'' pulses in a sequence of $n\gg 1$ pulses.  However, for ``nearly neighbouring pulses'' and sufficiently large number of fermions, so that $\vert i-j\vert \ll n$, we have
$\frac{4\pi ^2 (i-j)^2 }{n^2 }\ll 1$ and hence
\begin{align}
    e^{-\delta\omega^2 (\tau_i-\tau_j)^2}\approx 1-\frac{4\pi^2(i-j)^2}{n^2}\, , \qquad 
    (\hbox{for} \vert i-j\vert \ll n)\, .
\end{align}
Thus we conclude that we cannot reach the limit where all $r_{ij}=\delta_{ij}$ for long integration times even if the pulses are well separated in time.

Next consider short integration times, for which ${\cal T}\ll 1/\Delta\omega$, or 
$\Delta\Omega\gg 2\pi \Delta \omega$.  In this case
\begin{align}
     \delta\omega^2 (\tau_i-\tau_j)^2 &\sim  \frac{(2\pi \Delta\omega)^2}{\Delta\Omega^2}
    \frac{(i-j)^2 }{n^2 } \ll 1 , \qquad (\hbox{for } {\cal T}\ll 1/\Delta\omega, 
    \hbox{ any } \vert i-j\vert)\, ,
\end{align}
always.  Here again, it is not possible reach the limit where $r_{ij}=\delta_{ij}$.

Hence, we conclude that, for Fourier-limited detectors, we cannot reach a regime where all particles are considered fully distinguishable.
Instead we reach this regime by sending particles one by one at times $t_n=n\Delta T$ but detecting after a time $\approx (n+\frac{1}{2})\Delta T$; in this case the particles are
detected one by one, and never detected in coincidence.

One could instead imagine an alternate model where the $r_{ij}$'s are not measured by a temporal overlap, but by a polarization overlap, \emph{i.e.} working with polarized fermions.  One then easily shows that $r_{ij}=\cos^2(\theta_{ij})$, where $\theta_{ij}$ is the relative angle of the polarization vectors for particles $i$ and $j$.  In such a scenario, it is clearly not possible to have more that two particles which are pairwise fully distinguishable, and it is not possible to reach
the limit where all $r_{ij}=\delta_{ij}$ either.
}

\section{Discretizing time and Gamas's theorem}
\label{subsec:Gamas}
Although arrival times are continuous in principle,
and hence range over any real value 
in some interval,
in practice arrival times are discrete due to resolution limits of delay lines.
In this section,
we exploit discretization to write arrival times as a tuple of integers, with each integer in the tuple denoting the arrival time of the corresponding 
particle at the detector.
Absolute arrival times are not needed and we use {relative} arrival times to introduce the delay partition.
Next we explain how these partitions can be partially ordered,
which sets the stage for introducing and using Gamas's theorem.
Finally, we use these rules for partitioning to restrict sums used for computing the coincidence rate equations. Here again the discussion applies any
$s$ with $s_i=0$ or $1$, and we assume that no 
one particle is fully distinguishable from any of the others.

\subsection{Discretizing time and the delay partition }
\label{subsec:timebin}

In the remainder of the paper, we  use the combinatorics of Young diagrams to show that when some of the 
$\vert \bar{\tau}_i-\bar{\tau}_j\vert \approx 0$, 
some of the terms in Eqs.~(\ref{eq:rbn}) and~(\ref{eq:rfn}) are in turn automatically $0$, thus further simplifying the actual calculation by truncating the sum.

If the arrival times are controlled to arbitrary precision, the exact rate is given by Eq.~(\ref{eq:rateuRu}). In practice, we envisage a scenario where, if $\vert\bar\tau_i-\bar\tau_j\vert<\epsilon$, the incident particles are considered to arrive simultaneously and thus indistinguishable. It is then convenient to imagine that the source supplies, for each run,~$n$ particles arriving within some finite time interval $\mathcal{T}$.  This total time interval is divided into a number $b\ge 2$ of discrete identical time bins of with $\epsilon=\mathcal{T}/b$, so that a particle arriving at time $\bar \tau_i $ is assigned the discrete arrival time $\tau_c$
\begin{align}
\bar\tau_i \to \tau_c\qquad \hbox{when}\qquad \frac{(c-1) \mathcal{T}}{b}\le \bar\tau_i <
\frac{c\mathcal{T}}{b}\, ,\qquad c=1,\ldots, b\, .
\end{align}

The next step is construct a \emph{delay partition} ~$\mu_{\bm\tau}$ given a vector~$\bm\tau$ containing
discretized arrival times.
This is done by simply tallying
the number of particles in each bin, removing bins with no particles, and ordering the counts in decreasing order.  Thus, 
if {there are} $b=12$ bins, there are
$12$ possible discrete values of $\tau_i$ for each particle. 
If additionally there
are $n=4$ particles, the discrete arrival-time vector is  $\bm\tau=[\tau_1,\tau_4,\tau_1,\tau_9]$, we assign 
to this arrival-time vector $\bm\tau$ the delay partition $\mu_{\bm\tau}=(2,1,1)$.  

The delay partition is determined only by the tally in each bin, irrespective of 
the index of the (discretized) arrival times $\tau_i$ and irrespective of the actual
bin in which $\tau_i$ falls.  The (discrete) arrival-time vectors
\begin{align}
    [\tau_2,\tau_1,\tau_1,\tau_6]\, ,\qquad [\tau_7,\tau_3,\tau_4,\tau_3]\, ,\qquad 
    [\tau_5,\tau_5,\tau_2,\tau_{11}]
\end{align}
for instance all map to the same delay partition $\mu_{\bm\tau}=(2,1,1)$.  In this way a
delay partition is an equivalence class of discrete arrival-time vectors, 
and all the previous partitions are equivalent to the 
arrival-time vector $[\tau_1,\tau_1,\tau_2,\tau_3]$. For $n=4$ and $b\ge 4$, we have the following possible delay partitions:
 \begin{align}
     [\tau_1,\tau_1,\tau_1,\tau_1] &\rightarrow \mu_{\tau}= \Yvcentermath1\Yboxdim{4pt}\yng(4) \\ 
     [\tau_1,\tau_1,\tau_1,\tau_2] &\rightarrow \mu_{\tau}= \Yvcentermath1\Yboxdim{4pt}\yng(3,1) \nonumber \\ 
     [\tau_1,\tau_1,\tau_2,\tau_2] &\rightarrow \mu_{\tau}= \Yvcentermath1\Yboxdim{4pt}\yng(2,2) \nonumber \\ 
     [\tau_1,\tau_1,\tau_2,\tau_3] &\rightarrow \mu_{\tau}= \Yvcentermath1\Yboxdim{4pt}\yng(2,1,1) \nonumber \\
     [\tau_1,\tau_2,\tau_3,\tau_4] &\rightarrow \mu_{\tau}= \Yvcentermath1\Yboxdim{4pt}\yng(1,1,1,1) \nonumber
    \end{align}

\subsection{Dominance Ordering and Gamas's theorem}
\label{subsec:gamas}

For a given delay partition, one can obtain a significant simplification in the evaluation of the coincidence rate equations.  This simplification is a consequence of Gamas's theorem, from which one deduces that some immanants of a Gram matrix are zero based on dominance ordering of partitions.

Dominance ordering is a partial ordering of partitions,  given by $\mu \unlhd \lambda$ when
\begin{align}
\mu_1 + \mu_2 + \dots + \mu_i \leq \lambda_1 + \lambda_2 + \dots + \lambda_i
\label{eq:partialorder}
\end{align}
for all $i \geq 1$. We say that two partitions cannot be compared when neither one dominates the other. The notation $\mu \lhd \lambda$
means that $\lambda$ strictly dominates $\mu$ (a partition dominates itself, but it doesn't strictly dominate itself). An example is given in Fig.~\ref{fig:dominanceorder}, where it can be seen
that $\Yvcentermath1\Yboxdim{4pt}\yng(6)$ dominates all partitions of 6 and that 
$\Yvcentermath1\Yboxdim{4pt}\yng(3,3)$
and  $\Yvcentermath1\Yboxdim{4pt}\yng(4,1,1)$ cannot be compared.

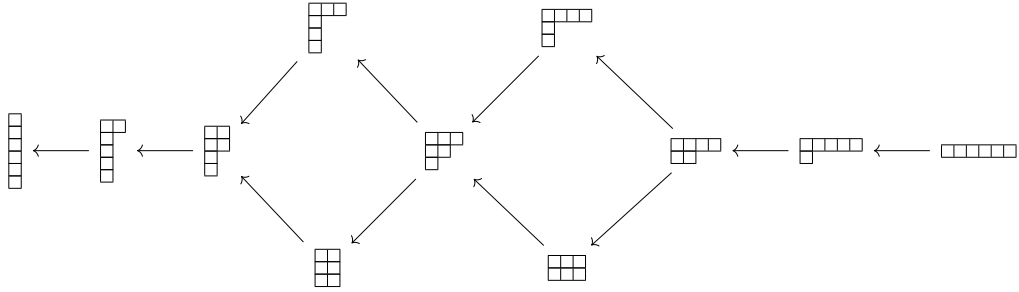
\begin{figure}[hbt!]
 \centering
 \begin{tikzcd}
     &  & & \Yvcentermath1\Yboxdim{5pt}\yng(3,1,1,1) \arrow[ld,]  & & \Yvcentermath1\Yboxdim{5pt}\yng(4,1,1) \arrow[ld]  & & &\\
\Yvcentermath1\Yboxdim{5pt}\yng(1,1,1,1,1,1)  & \Yvcentermath1\Yboxdim{5pt}\yng(2,1,1,1,1) \arrow[l] & \Yvcentermath1\Yboxdim{5pt}\yng(2,2,1,1) \arrow[l]  & & \Yvcentermath1\Yboxdim{5pt}\yng(3,2,1) \arrow[lu] \arrow[ld]  & & \Yvcentermath1\Yboxdim{5pt}\yng(4,2) \arrow[lu] \arrow[ld] & \Yvcentermath1\Yboxdim{5pt}\yng(5,1) \arrow[l] & \Yvcentermath1\Yboxdim{5pt}\yng(6) \arrow[l] \\
                             &  & &  \Yvcentermath1\Yboxdim{5pt}\yng(2,2,2) \arrow[lu]  & & \Yvcentermath1\Yboxdim{5pt}\yng(3,3) \arrow[lu]  & & &
\end{tikzcd}
 \caption{The dominance ordering of the partitions of $6$.}
 \label{fig:dominanceorder}
\end{figure}

\noindent Gamas's theorem tells us whether the $\lambda$-immanant of a Gram matrix is zero by looking at the shape of the partition $\lambda$. The theorem is stated as follows.

\begin{theorem}
(Gamas~\cite{gamas}): Let $V_{ij}=\bra{f_i} f_j \rangle$ be a $n \times n$ Gram matrix formed by the set of basis vectors $\{ f_i \}$, and let $\lambda \vdash n$ be a partition. We have that $\operatorname{imm}{^\lambda(V)}\neq 0$ if and only if the columns of $\lambda$ can partition the set of basis vectors $\{f_i\}$ into linearly independent sets.
\end{theorem}

In Table \ref{fig:GamasTable}, a $0$ indicates that the immanant, for
the corresponding set of basis functions, is $0$. Gamas's theorem is a 
bi-conditional statement, so the blank boxes mean the corresponding 
immanant is non-zero. For further reading on Gamas's theorem see~\cite{BergetGamas}.

\begin{table}[hbt!]
 \centering
\begin{tabular}{ |c|c|c|c|c|c|c|c|c|c|c|c| } 
   \hline
  $\Yvcentermath1\Yboxdim{4pt}\yng(1,1,1,1,1,1)$  &  $\Yvcentermath1\Yboxdim{4pt}\yng(2,1,1,1,1)$ & $\Yvcentermath1\Yboxdim{4pt}\yng(2,2,1,1)$ & 
  $\Yvcentermath1\Yboxdim{4pt}\yng(2,2,2)$ & 
  $\Yvcentermath1\Yboxdim{4pt}\yng(3,1,1,1)$ & 
  $\Yvcentermath1\Yboxdim{4pt}\yng(3,2,1)$ & $\Yvcentermath1\Yboxdim{4pt}\yng(3,3)$ & 
  $\Yvcentermath1\Yboxdim{4pt}\yng(4,1,1)$ & 
  $\Yvcentermath1\Yboxdim{4pt}\yng(4,2)$ & 
  $\Yvcentermath1\Yboxdim{4pt}\yng(5,1)$ & 
  $\Yvcentermath1\Yboxdim{4pt}\yng(6)$ & \phantom{ $\Yvcentermath1\Yboxdim{4pt}\yng(1,1,1,1,1,1,1,1)$} \\
    \hline
   & & & & & & & & & & \phantom{$\Yvcentermath1\Yboxdim{5pt}\yng(1,1,1)$} & $f_1$ $f_2$ $f_3$ $f_4$ $f_5$ $f_6$ \\
  \hline
  0 & & & & & & & & & & \phantom{$\Yvcentermath1\Yboxdim{5pt}\yng(1,1,1)$} & $f_1$ $f_1$ $f_2$ $f_3$ $f_4$ $f_5$ \\
   \hline
  0 & 0 & & & & & & & & & \phantom{$\Yvcentermath1\Yboxdim{5pt}\yng(1,1,1)$} & $f_1$ $f_1$ $f_2$ $f_2$ $f_3$ $f_4$ \\
   \hline
   0 & 0 & 0 & & 0 & & & & & & \phantom{$\Yvcentermath1\Yboxdim{5pt}\yng(1,1,1)$} & $f_1$ $f_1$ $f_2$ $f_2$ $f_3$ $f_3$  \\
   \hline
   0 & 0 & 0 & 0 & & & & & & & \phantom{$\Yvcentermath1\Yboxdim{5pt}\yng(1,1,1)$} & $f_1$ $f_1$ $f_1$ $f_2$ $f_3$ $f_4$ \\
   \hline
   0 & 0 & 0 & 0 & 0 & & & & & & \phantom{$\Yvcentermath1\Yboxdim{5pt}\yng(1,1,1)$} & $f_1$ $f_1$ $f_1$ $f_2$ $f_2$ $f_3$ \\
   \hline
     0 & 0 & 0 & 0 & 0 & 0 & & 0 & & & \phantom{$\Yvcentermath1\Yboxdim{5pt}\yng(1,1,1)$} & $f_1$ $f_1$ $f_1$ $f_2$ $f_2$ $f_2$ \\
   \hline
     0 & 0 & 0 & 0 & 0 & 0 & 0 & & & & \phantom{$\Yvcentermath1\Yboxdim{5pt}\yng(1,1,1)$} & $f_1$ $f_1$ $f_1$ $f_1$ $f_2$  $f_3$ \\
   \hline
    0 & 0 & 0 & 0 & 0 & 0 & 0 & 0 & & & \phantom{$\Yvcentermath1\Yboxdim{5pt}\yng(1,1,1)$} & $f_1$ $f_1$ $f_1$ $f_1$ $f_2$ $f_2$ \\
   \hline
    0 & 0 & 0 & 0 & 0 & 0 & 0 & 0 & 0 & & \phantom{$\Yvcentermath1\Yboxdim{5pt}\yng(1,1,1)$} & $f_1$ $f_1$ $f_1$ $f_1$ $f_1$ $f_2$ \\
   \hline 
    0 & 0 & 0 & 0 & 0 & 0 & 0 & 0 & 0 & 0 & \phantom{$\Yvcentermath1\Yboxdim{5pt}\yng(1,1,1)$} & $f_1$ $f_1$ $f_1$ $f_1$ $f_1$ $f_1$ \\
   \hline
\end{tabular} 
 \caption{The vanishing behaviour of the immanants of $6 \times 6$ Gram matrices.}
 \label{fig:GamasTable}
\end{table}

\subsection{Restricting the sum in the coincidence rate equations.}
\label{sec:restricting}

Applying a corollary of Gamas's theorem ~\cite{silva} to Lemma  \ref{lemma:zeroentries}
yields the following: 
\begin{proposition}
\label{coincidencerates}
Assume $s$ with $s_i=0$ or $1$. The matrix 
$\mathfrak{R}^\lambda(\bm\tau)$ is non-zero if and only if $\lambda$ dominates $\mu_{\tau}$ in the dominance ordering of partitions. Thus, for a given set of {discrete} arrival times $\bm\tau$, 
the coincidence rate equations simplify to
\begin{align}
\operatorname{rate}(\bm\tau,s;\text{b})
=&\sum_{\mu_{\tau} \unlhd \lambda}
\sum_{i=1}^{\text{s}_\lambda} \mathfrak{v}^\dagger_{\lambda;i}(s)\, 
\mathfrak{R}^\lambda(\bm\tau)\, \mathfrak{v}_{\lambda;i}(s),\\
\operatorname{rate}(\bm\tau,s;\text{f}) =
&\sum_{\mu_{\tau} \unlhd \lambda} 
\sum_{i=1}^{\text{s}_\lambda} \mathfrak{v}^\dagger_{\lambda;i}(s)\, 
\mathfrak{R}^{\lambda^*}(\bm\tau)\, \mathfrak{v}_{\lambda;i}(s). \label{eq:fermrate}
\end{align}
The sum is over all partitions $\lambda \vdash n$ that dominate $\mu_{\tau}$.
\end{proposition}\\
\noindent
We note that if every particle arrives in a different time bin, there is no truncation and we recover the full
sums of Eqs.(\ref{eq:rbn}) and~(\ref{eq:rfn}).

\begin{example}
 Consider the set of {discrete} arrival times $\bm\tau=[\tau_1,\tau_1,\tau_1,\tau_2,\tau_2,\tau_2 ]$, corresponding to the delay partition  $\mu_{\tau}=\Yvcentermath1\Yboxdim{3pt}\yng(3,3)$. It follows that for $\lambda$ equal to $\,\Yvcentermath1\Yboxdim{4pt}\yng(1,1,1,1,1,1),\, \Yvcentermath1\Yboxdim{4pt}\yng(2,1,1,1,1),\, \Yvcentermath1\Yboxdim{4pt}\yng(2,2,1,1),\, 
 \Yvcentermath1\Yboxdim{4pt}\yng(3,1,1,1),\, \Yvcentermath1\Yboxdim{4pt}\yng(2,2,2),\,\Yvcentermath1\Yboxdim{4pt}\yng(3,2,1),$ and
 $\Yvcentermath1\Yboxdim{4pt}\yng(4,1,1)$ the $\mathfrak{R}^{\lambda}(\bm\tau)$ terms are equal to zero since these are the partitions which do not dominate $\Yvcentermath1\Yboxdim{3pt}\yng(3,3)$. 
 The coincidence rate equations are thus
\begin{align}
    \operatorname{rate}(\bm\tau,s;\text{b}) =& \mathfrak{v}^\dagger_{\Yvcentermath1\Yboxdim{3pt}\yng(6);1}(s) \mathfrak{R}^{\Yvcentermath1\Yboxdim{3pt}\yng(6)}(\bm\tau) \mathfrak{v}_{\Yvcentermath1\Yboxdim{3pt}\yng(6);1}(s) + \sum_{i=1}^5 \mathfrak{v}^\dagger_{\Yvcentermath1\Yboxdim{3pt}\yng(5,1);i}(s) \mathfrak{R}^{\Yvcentermath1\Yboxdim{3pt}\yng(5,1)}(\bm\tau) \mathfrak{v}_{\Yvcentermath1\Yboxdim{3pt}\yng(5,1);i}(s) \nonumber \\
    &\qquad + \sum_{i=1}^9 \mathfrak{v}^\dagger_{\Yvcentermath1\Yboxdim{3pt}\yng(4,2);i}(s) \mathfrak{R}^{\Yvcentermath1\Yboxdim{3pt}\yng(4,2)}(\bm\tau) \mathfrak{v}_{\Yvcentermath1\Yboxdim{3pt}\yng(4,2);i}(s)+ \sum_{i=1}^5 \mathfrak{v}^\dagger_{\Yvcentermath1\Yboxdim{3pt}\yng(3,3);i}(s) \mathfrak{R}^{\Yvcentermath1\Yboxdim{3pt}\yng(3,3)}(\bm\tau) \mathfrak{v}_{\Yvcentermath1\Yboxdim{3pt}\yng(3,3);i}(s), \\
     \operatorname{rate}(\bm\tau,s;\text{f}) =& \mathfrak{v}^\dagger_{\Yvcentermath1\Yboxdim{3pt}\yng(1,1,1,1,1,1);1}(s) \mathfrak{R}^{\Yvcentermath1\Yboxdim{3pt}\yng(6)}(\bm\tau) \mathfrak{v}_{\Yvcentermath1\Yboxdim{3pt}\yng(1,1,1,1,1,1);1}(s) + \sum_{i=1}^5 \mathfrak{v}^\dagger_{\Yvcentermath1\Yboxdim{3pt}\yng(2,1,1,1,1);i}(s) \mathfrak{R}^{\Yvcentermath1\Yboxdim{3pt}\yng(5,1)}(\bm\tau) \mathfrak{v}_{\Yvcentermath1\Yboxdim{3pt}\yng(2,1,1,1,1);i}(s) \nonumber \\
    &\qquad + \sum_{i=1}^9 \mathfrak{v}^\dagger_{\Yvcentermath1\Yboxdim{3pt}\yng(2,2,1,1);i}(s) \mathfrak{R}^{\Yvcentermath1\Yboxdim{3pt}\yng(4,2)}(\bm\tau) \mathfrak{v}_{\Yvcentermath1\Yboxdim{3pt}\yng(2,2,1,1);i}(s)+ \sum_{i=1}^5 \mathfrak{v}^\dagger_{\Yvcentermath1\Yboxdim{3pt}\yng(2,2,2);i}(s) \mathfrak{R}^{\Yvcentermath1\Yboxdim{3pt}\yng(3,3)}(\bm\tau) \mathfrak{v}_{\Yvcentermath1\Yboxdim{3pt}\yng(2,2,2);i}(s). 
\end{align}
\end{example}

\section{Fermionic rates and partial distinguishability}

For fully indistinguishable fermions, the coincidence rate is given by the first term in Eq.~(\ref{eq:rfn}) and contains only the modulus squared of the determinant of the matrix $A(s)$; for this case the matrix $r(\bm{\bar{\tau}})$ is the unit matrix and the permanent of $R(\bm{\tau})=n!$.  Since determinants can be evaluated efficiently, the evaluation of 
fermionic rates in this case is simple.

As we increase the distinguishability of the fermions, the 
expression for the coincidence rate lengthens to contain an increasing number of group functions, as per Proposition~\ref{coincidencerates}. Moreover, which group functions occur is determined by Gamas's theorem and the delay partition obtained from the times of arrival.

In this section we show that, once we reach a situation where at most $\ceil*{\frac{n}{2}}$ fermions are 
pairwise indistinguishable, the expression for the coincidence rate will contain specific group functions that are evaluated, using the algorithm of \cite{burgisser1}, in a number of operations that grows exponentially with the number of fermions. We further show that, if the times of arrival are uniformly random, one must evaluate these expensive groups functions with probability differing from $1$ by a number decreasing exponentially with the number of fermions. In this section we suppose again that $s$ is such that $s_i=0$ or $1$.

\subsection{The Witness Partition}\label{subsec:witnesscomplexity}

We begin this section by defining the witness partition and proving a lemma.

\begin{definition}
The witness partition $\mu_w$ is 
\begin{align}
    \mu_w:=\left\{ \begin{array}{lc} 
    \left(\displaystyle\frac{n}{2},\frac{n}{2}\right)&\hbox{for~$n$ even},\\
    \\
    \left(\displaystyle\frac{n+1}{2},\frac{n-1}{2}\right)&\hbox{for~$n$ odd}.
    \end{array}
    \right.
\end{align}
\end{definition}

\noindent We will often refer to the conjugate of the witness partition, which is 

\begin{align}
    \mu^*_w=\left\{ \begin{array}{lc} 
    \left(2,2,\dots,2\right)&\hbox{for~$n$ even},\\
    \\
    \left(2,2,\dots,2,1\right)&\hbox{for~$n$ odd}.
    \end{array}
    \right.
\end{align}

Consider the case of fermionic interferometry. We present the following lemma.  

\begin{lemma}
When $n$ is even, group functions of type $\mu^*_w$ of the matrix $A(s)$ need to be computed to evaluate the fermionic coincidence rate equation if and only if at most $\frac{n}{2}$ of the times of arrival are contained in a single time bin, and we assume that no 
one particle is fully distinguishable from any of the others.

When $n$ is odd, group functions of type $\mu^*_w$ of the matrix $A(s)$ need to be computed to evaluate the fermionic coincidence rate equation if and only if at most $\frac{n+1}{2}$ of the times of arrival are contained in a single time bin, and we assume that no one particle is fully distinguishable from any of the others.
\end{lemma}

\begin{proof}
From Eq.~(\ref{eq:fermrate}), we observe that if $\mu_{\tau} \unlhd \lambda$, then group functions of type $\lambda^*$ of the matrix $A(s)$ appear in the fermionic coincidence rate equation; otherwise, we have that when $\lambda$ \emph{doesn't} dominate $\mu_{\tau}$ (meaning that either $\lambda \lhd \mu_{\tau}$ or the two partitions cannot be compared), then group functions of type $\lambda^*$ of the matrix $A(s)$ \emph{do not} appear in the fermionic coincidence rate equation.

First, consider the case where $n$ is even and assume that group functions of type $\mu^*_w$ appear in the fermionic rate equation. It follows that $\mu_w$ dominates the delay partition: $\mu_{\tau} \unlhd \mu_w$. The partition $\mu_w$ is of the form $\left(\textstyle\frac{n}{2},\frac{n}{2}\right)$, thus $\mu_w$ dominates all partitions of width at most $\frac{n}{2}$. By construction, the partitions of width at most $\frac{n}{2}$ correspond to the instances where at most $\frac{n}{2}$ fermions are in a single time bin. Each implication in our argument is bi-conditional, so the result for even $n$ follows.

In the case where $n$ is odd, $\mu_w$ is now of the form $\left(\textstyle\frac{n+1}{2},\frac{n-1}{2}\right)$. It's clear that $\mu_w$ dominates all of the partitions of width at most $\frac{n+1}{2}$. The rest of the proof is identical to the even case.
\end{proof}

Note that regardless of whether the number of fermions is odd or even, if at most $\ceil*{\frac{n}{2}}$ of the times of arrival are in the same time bin, then group functions of type $\mu^*_w$ of the matrix $A(s)$ need to be evaluated to determine the fermionic rate. 

We proceed by showing that  \emph{some} group functions of type $\mu^*_w$ are computationally expensive to evaluate. B\"{u}rgisser~\cite{burgisser1} shows that the number of arithmetic operations needed to evaluate the immanant of Eq.~(\ref{eq:immassumofD}) in the $\lambda \vdash n$ representation of $GL_n(\mathbb{C})$ is $\mathcal{O}\left([\operatorname{mult}(\lambda)+\operatorname{log(n)}]n^2 s_\lambda\text{d}_\lambda \right).$ The sum in Eq.~(\ref{eq:immassumofD}) contains exactly $s_\lambda$ terms, so at least one of the $\mathcal{D}^{\lambda}_{i,i}(A(s))$ must evaluate in at least
\begin{align}
\label{eq: Dfunctioncomplexity}
\mathcal{O}\left([\operatorname{mult}(\lambda)+\operatorname{log(n)}]n^2 \text{d}_\lambda \right).    
\end{align}
operations. Here, $\text{s}_\lambda$ is both the number of standard Young
tableaux of shape $\lambda$ and the dimension of the weight $(1,1,1,\ldots,1)$
subspace of the $\lambda$-irrep of $GL_n(\mathbb{C})$ and $\text{d}_\lambda$ is both
the number of semi-standard Young tableaux of shape $\lambda$ and the
dimension of the $\lambda$-irrep of $GL_n(\mathbb{C})$,
$\operatorname{mult}(\lambda)$ is the multiplicity of the highest weight
$\lambda$~\cite{burgisser1}. Furthermore, \emph{any} 
arithmetic algorithm that evaluates group functions of type $\lambda$ ``requires at least $\hbox{d}_{\lambda}$ nonscalar operations'' according Theorem $5.1$ of~\cite{burgisser1}.

We now compute $d_{\mu^*_w}$ to get a lower bound for Eq.~(\ref{eq: Dfunctioncomplexity}) .
Suppose~$n$ is even so that $\mu^*_{w}=(2,2,\ldots,2)$. Exponents are used to denote repeated entries, so the partition can be more compactly written as $\mu^*_{w}=(2^{n/2})$.  
Then one easily shows that $s_{(2^{n/2})}$ is given by the Catalan number
\begin{align}
    C_{n/2}=\frac{1}{\frac{n}{2}+1}{n \choose \frac{n}{2}}\sim 
    \frac{2^{n+3/2}}{\sqrt{\pi n^3}}
\end{align}
for large~$n$ so the number of terms to evaluate is growing exponentially with the number of fermions. Moreover, the dimension of the irrep
$d_{(2^{n/2})}$ of $GL_n(\mathbb{C})$
is given by
\begin{align}
    d_{(2^{n/2})}=
    {n \choose \frac{n}{2}} {n+1\choose \frac{n}{2}}\frac{1}{\frac{n}{2}+1}\
    =C_{n/2} {n+1\choose \frac{n}{2}}\, .
\end{align}
The binomial coefficient ${n+1\choose \frac{n}{2}}\sim 2^{n+3/2}e^{-1/(2n+2)}/\sqrt{(n+1)\pi}$ for
large~$n$ so that, altogether:
\begin{align}
    d_{(2^{n/2})}\approx \frac{2^{2n+3}}{n^2\pi}
\end{align}
also scales exponentially with the number of fermions.  Thus, at least one
of the group functions necessary to compute the fermionic 
coincidence rate equation requires a number
of arithmetic operations that scales exponentially with the number of fermions.
A similar expression and argument can be made when~$n$ is odd. Thus, combining this with Theorem 5.1 of \cite{burgisser1} as discussed above, we have thus shown the following.
\begin{proposition}
\label{proposition:hardness}
If at most $\ceil*{\frac{n}{2}}$ of the times of arrival are contained in any one time bin, {and no particle is fully distinguishable from any other,} then our procedure for the exact computation of the fermionic coincidence rate equation requires a number of arithmetic operations that scales exponentially in the
number of fermions.
\end{proposition}



\subsection{Probability of the witness partition appearing for uniformly-random arrival times}
\label{subsec:p}

In this subsection, we discuss a simple model where
the discretized arrival times are uniformly random over some interval $\mathcal{T}$. Fix some string $s$ where $s_i=0$ or $1$.  Imagine~$n$ possible uniformly random discrete arrival times (one per particle) distributed over a total of $b$ time bins. The arrival times can be repeated. Given a delay partition $\mu_\tau=(\mu_1,\mu_2,\ldots,\mu_k)$ , 
we need to find how many ways we can distribute the~$n$ particles in $b$ bins so that $\mu_1$ are in 
any one bin, $\mu_2$ are in any other bin which is not the first, \emph{etc}. 

For instance, if $n=6$ and $b=8$, then the arrival times could be
$\bm\tau=[1,5,8,5,8,3]$. For clarity we tabulate the results:
\begin{table}[H]
    \centering
    {\renewcommand{\arraystretch}{1.5}
    \begin{tabular}{|c|c|c|c|c|c|c|} \hline 
        particle \# & \,1\,&\,2\,&\,3\,&\,4\, &\,5 \,& \,6\,\\
        \hline 
        bin \# & 1& 5&8 & 5&8&3 \\ 
        \hline 
    \end{tabular}}
    \caption{Tabulated results of a hypothetical experiment.}
    \label{tab:results}
\end{table}

\noindent This arrival time vector $\bm{\tau}$ is associated to the delay partition $\mu_\tau=(2,2,1,1)$. Let $b_i$ denote the number of time bins that
contain exactly $i$ particles. In this example, the fifth and eighth time
bins contain two particles each, so $b_2=2$; the first and third time bins
contain one particle, so $b_1=2$; the second, fourth, sixth, and seventh
time bins contain no particles, so $b_0=4$. In general, the $b_i$ terms
 satisfy the following constraints 
\begin{align}
\sum_{i=0}^n b_i =& b \\
\sum_{i=0}^n i b_i =& n. \nonumber
\end{align}

With $b$ time bins and~$n$ particles there are $b^n$ possible ways for the particles to arrive at the detectors; assuming that each of these possibilities are equally likely, we present an expression for the probability that a random input of particles has a particular delay partition.

\begin{proposition}
\label{proposition:probability}
In an experiment with~$n$ particles and $b$ time bins, the probability that a uniformly-random set of arrival times has delay partition $\mu_{\tau}=(\mu_1, \mu_2, \dots, \mu_k)$ is
\begin{align}
\mathcal{P}(\mu_{\tau};b)=\frac{1}{b^n} \binom{n}{\mu_1, \mu_2, \dots, \mu_k} \binom{b}{b_0, b_1, \dots,  b_n},  \label{eq:delayprob}  
\end{align}
where the $b_i$'s are computed from the delay partition $\mu_\tau$ as described
above.  If $b$ is less than the number of parts $k$ of a delay
partition, then the probability of that delay partition occurring is $0$.
\end{proposition}


\begin{example}
Suppose there are $8$ time bins, $5$ particles, and we are computing the
probability of the delay partition $\mu_{\tau}=(2,2,1)$ occurring. Of the
$8$ time bins, there is $1$ bin that contains $1$ particle, $2$ bins that
contain $2$ particles, and $5$ bins that contain $0$ particles; thus,
$b_0=5, b_1=1, b_2=2$. Applying Proposition \ref{proposition:probability}, we get 

\begin{align}
\mathcal{P}((2,2,1);8) =&\frac{1}{8^5} \frac{5!}{(2!)(2!)(1!)}  \frac{8!}{(5!)(1!)(2!)} \\
=& \frac{315}{2048}. \nonumber
\end{align} 
\end{example}

Recall that Proposition \ref{proposition:hardness} states that when at most $\ceil*{\frac{n}{2}}$ of the times of arrival are contained in a single time bin, this means that our method for computing the fermionic coincidence rate equation requires exponentially many arithmetic operations. Let $\mathcal{P}$ denote the probability of obtaining a delay partition $\mu_{\tau}=\left(\mu_1,\mu_2,\dots \right)$ with $\mu_1 \leq \ceil*{\frac{n}{2}}$. We show that $1 - \mathcal{P}$, which is the probability of obtaining a delay partition with $\mu_1 > \ceil*{\frac{n}{2}}$, approaches $0$ exponentially fast with $n$.

We begin with the case where $n$ is even. Let $P(B)$ be the probability of getting a delay partition with first part equal to $\lambda_1$ and the remaining parts are arbitrary, and let $P(A)$ be the probability of obtaining delay partition $\lambda'=(\lambda_2, \lambda_3, \dots, \lambda_k) \vdash n-\lambda_1$ for the remaining $n-\lambda_1$ particles being sorted into the remaining $(b-1)$ bins. We get that the conditional probability $P(A|B)$ is given by

\begin{align}
P(A|B)=\frac{1}{(b-1)^{n-\lambda_1}} \binom{n-\lambda_1}{\lambda_2, \lambda_3, \dots, \lambda_k} \binom{b-1}{b_0, b_1, \dots, b_{\lambda_1-1},{1}, b_{\lambda_1+1}, \dots  b_n},    
\end{align}
{where} $b_{\lambda_1}=1$. The probability $P(A \cap B)$ is simply the probability of obtaining an arbitrary delay partition $\lambda$, which is given in Eq.~(\ref{eq:delayprob}). The probability $P(B)$ is thus given by the following quotient 

\begin{align}
\frac{P(A \cap B)}{P(A|B)} &= 
\frac{\frac{1}{b^n} \binom{n}{\lambda_1, \lambda_2, \dots, \lambda_k} \binom{b}{b_0, b_1, \dots,  b_n}}{\frac{1}{(b-1)^{n-\lambda_1}} \binom{n-\lambda_1}{\lambda_2, \lambda_3, \dots, \lambda_k} \binom{b-1}{b_0, b_1, \dots, b_{\lambda_1-1}, {1},b_{\lambda_1+1}, \dots  b_n}}, \nonumber \\
&= \frac{(b-1)^{n-\lambda_1}}{b^n} \frac{n!}{\lambda_1! (n-\lambda_1)!} \frac{b!}{(b-1)!}, \nonumber \\
&= \binom{n}{\lambda_1} \frac{(b-1)^{n-\lambda_1}}{b^{n-1}}.
\end{align}

\noindent To obtain $1 - \mathcal{P}$ we need to sum over all partitions $\lambda \vdash n$ with $\lambda_1 \geq \frac{n}{2}+1$,

\begin{align}
1 - \mathcal{P} = \frac{1}{b^{n-1}} \sum_{\lambda_1 = \frac{n}{2}+1}^n \binom{n}{\lambda_1} (b-1)^{n-\lambda_1}.
\end{align}

\noindent We factor out $(b-1)^n$ from each term in the sum to get 

\begin{align}
1 - \mathcal{P} = \frac{(b-1)^n}{b^{n-1}} \sum_{\lambda_1 = \frac{n}{2}+1}^n \binom{n}{\lambda_1} \left(\frac{1}{b-1}\right)^{\lambda_1}.
\end{align}

\noindent We assume that $n$ and $b$ are large to be able to truncate the sum. We have 

\begin{align}
1 - \mathcal{P} &= \frac{(b-1)^{n-1}}{b^{n-1}}
    \left({n\choose n/2+1}\left(\frac{1}{b-1}\right)^{n/2}
    + {n\choose n/2+2}\left(\frac{1}{b-1}\right)^{1+n/2}
   +\ldots
    \right)\, , \nonumber \\
       &=\frac{(b-1)^{n-1}}{b^{n-1}}{n\choose n/2+1}
   \left(\frac{1}{b-1}\right)^{n/2}
    \left(1+ \frac{n-2}{n+4}
     \left(\frac{1}{b-1}\right)
   +\ldots \right)\, \nonumber ,\\
   & \approx \frac{(b-1)^{n-1}}{b^{n-1}} \frac{2^{n+\frac{1}{2}}}{\sqrt{n \pi}} \left(\frac{1}{b-1}\right)^{n/2}
    \left(1+ \frac{n-2}{n+4}
     \left(\frac{1}{b-1}\right)
   +\ldots \right).
\end{align}
By assuming {again} that $b-1 \approx b$ and by truncating the sum after the first term, we find that 
\begin{align}
1 - \mathcal{P} \approx \sqrt{\frac{2}{n \pi}} \left( \frac{4}{b}\right)^{\frac{n}{2}}.
\end{align}

\noindent We note that when $n$ is odd there is a near-identical expansion and we get the same estimate for the probability.

\begin{corollary}
In an experiment with~$n$ particles that arrive uniformly randomly in $b$ discrete time bins, the probability $\mathcal{P}$ of at most $\ceil*{\frac{n}{2}}$ particles being in the same time bin is
\begin{align}
\mathcal{P} \approx 1 - \sqrt{\frac{2}{n \pi}} \left( \frac{4}{b}\right)^{\frac{n}{2}}.     
\end{align}
\end{corollary}

\noindent This is also the probability of needing the compute the (expensive) $\mu^*_w$-group functions of $A(s)$. We observe that when $b \geq 5$, then $\mathcal{P}$ approaches $1$ exponentially fast as~$n$ increases.

\section{Applications to Sampling Problems}

\subsection{Generalized Boson and Fermion Sampling}

Quantum computing focuses on technology and algorithms for solving certain computational problems more efficiently than what can be accomplished using non-quantum (i.e., ``classical'') computing,
essentially based on the binary representation of information and Boolean logic~\cite{Knu77}.
Both universal and specialized quantum computing approaches are followed. Gate-based quantum computing~\cite{NC11,San17}
is an example of universal quantum computing.
Both quantum annealing~\cite{HRIT15} and boson sampling~\cite{AA11,BGC+19} are examples of specialized, purposed quantum computing that is not universal.
Boson sampling is about simultaneously firing~$n$ single photons
so they arrive simultaneously at detectors behind an $m$-channel passive optical interferometer. 

\textsc{BosonSampling} is a hard classical problem and easy quantum problem in both exact and approximate formulations subject to assumptions. Recall from \S\ref{subsec:interferometry} that $\Phi_{m,n}$ is the set of $m$-tuples $(s_1,\ldots,s_m)$ so that $s_1+\ldots s_m=n$. When the bosons are indistinguishable, the probability distribution 
$\mathcal{B}[{\mathcal{A}}]$ over $\Phi_{m,n}$ is given by~\cite{AA11} 
\begin{align}
\label{def:bosonsampling}
\underset{\mathcal{B}[{\mathcal{A}}]}{\hbox{Pr}[s]} = \frac{|\operatorname{per}(A(s))|^2}{s_1!s_2!\dots s_m!}.
\end{align}

\begin{definition}
Given $\mathcal{A}$ as input, the problem of \textsc{BosonSampling} is to sample, either exactly or approximately, from the distribution $\mathcal{B}[{\mathcal{A}}]$. 
\end{definition}

The exact \textsc{BosonSampling} problem is not efficiently solvable by a classical computer, unless  the polynomial hierarchy collapses to the third level, which is an unlikely scenario based on existing results in computational complexity~\cite{AA11}.
 Despite the simplicity behind the complexity of \textsc{BosonSampling} in 
theory, experimental \textsc{BosonSampling} has limitations on size and on validity
of assumptions needed to trust hardness results. For example, the
assumptions on the perfectly indistinguishable photons are
impossible to
achieve in experiments: not only are photons distinguishable in their
temporal and frequency profile due to imperfect single-photon source,
experimental optical interferometers also introduce additional added
imperfections due to unequal path length, or unwanted phase errors. Notice
that such experimental imperfections, are different from the approximate
\textsc{BosonSampling}, which permits the sampling to deviate from the
perfect sampling of an ideal bosonic interferometer with indistinguishable
photon sources.  The computational hardness of imperfect
\textsc{BosonSampling} is still an ongoing subject of
discussion~\cite{PhysRevA.85.022332}.

When fermions are indistinguishable, the probability distribution $\mathcal{F}[{\mathcal{A}}]$ over $\Phi_{m,n}$ is given, up to a constant, by~\cite{AA13}
\begin{align}
\underset{\mathcal{F}[{\mathcal{A}}]}{\hbox{Pr}[s]} = \left|\operatorname{det}(A(s))\right|^2.
\end{align}
\noindent As a result of the Pauli exclusion 
principle,
two indistinguishable fermions cannot exit from the same channel
so the $s_1 ! s_2 ! \dots s_m !$ factor 
in the denominator of Eq.~(\ref{def:bosonsampling}) is replaced here by~$1$.
From the notion of fermion interferometry,
we define the associated problem of fermion sampling.
\begin{definition}
Given $\mathcal{A}$ as input, the problem of \textsc{FermionSampling} is to sample, either exactly or approximately, from the distribution $\mathcal{F}[{\mathcal{A}}]$. 
\end{definition}

\noindent The key application of \textsc{FermionSampling} is solving the same types of computational problems as for \textsc{BosonSampling} by opening the door to integrated semiconductor circuitry as an alternative to scaling challenges inherent in photonic approaches \cite{rudolph2017optimistic, ODMZ20}.

For simultaneous arrival times at the detectors,
fermion coincidence rates are easy to calculate~\cite{TD02,DT05}.
As a result, fermion interferometry has not been treated as a viable contender for quantum advantage~\cite{Pre12} since \textsc{FermionSampling} can be solved in classical polynomial time~\cite{AA13}
(unless a quantum resource is added~\cite{ODMZ20}).
Mathematically,
the ease of \textsc{FermionSampling} arises because the hard-to-solve
expressions for matrix permanents arising in boson coincidence calculations are replaced by matrix determinants for fermion coincidences~\cite{dGS17}.  The trivial nature of quantum vs classical algorithms for \textsc{FermionSampling} should 
not discourage exploiting fermion sampling for quantum computing provided that we 
generalize to nonsimultaneous arrival times.

In fact, in view of Proposition~\ref{proposition:hardness} the lack of simultaneity removes, in part or in totality, the argument
that fermion sampling is uninteresting because it is efficiently
simulatable. Indeed Proposition~\ref{proposition:hardness} suggests that, if we are to use the state-of-the-art algorithm to evaluate group functions and immanants, the evaluation of coincidence rates is exponentially expensive under reasonable conditions as to the maximum number of indistinguishable fermions.  However, while the series of Eq.~(\ref{eq:fermrate}) for the \emph{exact} rate will contain expensive group functions, we cannot guarantee that these function will have a significant contribution to the final rate.

Heyfron \cite{heyfron1988immant,heyfron1991some} and separately Pate \cite{pate1991partitions} have obtained some results on immanant inequalities which follow dominance ordering for positive semi-definite matrices, like the delay matrix $r(\boldsymbol{\tau})$ of Eq.~(\ref{eq:delaymatrix}), but these results don't encompass the witness partition. Stembridge \cite{stembridge1991immanants} has also shown that immanants of totally positive matrices, like our delay matrix, are necessarily positive (see also the appendix of \cite{huber2021matrix} for some immanant inequalities).  Nevertheless, to understand the contribution to the rates from group functions of $A(s)$, one would also need 
"anticoncentration"-type results for these functions or associated immanants for these expensive functions or immanants.

Nevertheless, as a result of Proposition~\ref{proposition:hardness}, it is natural to generalize both \textsc{BosonSampling} and \textsc{FermionSampling} 
by allowing nonsimultaneous arrival times specified by the
vector $\bar{\bm\tau}$. 
The set $G_{m,n}\subset \Phi_{m,n}$ contains only the strings $s$ such
that $s_i=0$ or $1$ for all $i$. The distribution
$\mathcal{B}\left[{\mathcal{A};\bar{\bm\tau}}\right]$ over $G_{m,n}$ is
defined by
\begin{align}
    \underset{\mathcal{B}\left[{\mathcal{A};\bar{\bm\tau}}\right]}{\hbox{Pr}[s]} = \operatorname{rate}(\bar{\bm \tau},s;\text{b})
\end{align}
where $\operatorname{rate}(\bar{\bm \tau},s;\text{b})$
is given in Eq.~(\ref{eq:rateuRu}).

\begin{problem}[Exact \textsc{GenBosonSampling}]
Given the $m \times n$ matrix $\mathcal{A}$ and a length~$n$ arrival-time vector
$\bar{\bm\tau}\in\mathbb{R}^n$ as inputs,
the problem of \textsc{GenBosonSampling} is to sample exactly from the distribution $\mathcal{B}[{\mathcal{A};\bar{\bm\tau}]}$.
\end{problem}

In the case of fermions, the distribution $\mathcal{F}\left[{\mathcal{A};\bar{\bm\tau}}\right]$ over $G_{m,n}$ is defined by

 \begin{align}
    \underset{\mathcal{F}\left[{\mathcal{A};\bar{\bm\tau}}\right]}{\hbox{Pr}[s]} = \operatorname{rate}(\bar{\bm \tau},s;\text{f})
\end{align}
where 
$\operatorname{rate}(\bar{\bm \tau},s;\text{f})$
is given in Eq.~(\ref{eq:fermionrateuRu}).

\begin{problem}[Exact \textsc{GenFermionSampling}]\label{def:Gfermionsampling}
Given the $m \times n$ matrix $\mathcal{A}$ and a length~$n$ arrival-time vector
$\bar{\bm\tau}\in\mathbb{R}^n$ as inputs, the problem of \textsc{GenFermionSampling} is to sample exactly from the distribution $\mathcal{F}[{\mathcal{A};\bar{\bm\tau}]}$.
\end{problem}

Comparing with the original definition of \textsc{BosonSampling} and
\textsc{FermionSampling}, the exact generalized definitions above admit, in
addition to the inputs~$n$ for the number of fermions, 
the arrival-time vector
$\bar{\bm\tau}\in\mathbb{R}^n$
but with the restriction of $s$ to strings with $s_i=1$ or $0$, and 
the proviso that real- and complex-number entries are admitted as floating-point numbers up to machine precision,
which is important to note because we are focused on \emph{exact} computation.
This definition differs from the one proposed in
\cite{ODMZ20} where simultaneity remains but the input state is no longer a single product 
state.

\subsection{Calibration and Fermion Sampling}
\label{subsec:calibration}

Every interferometry experiment has the problem of non-simultaneity due to ``lengths \dots not [being] well-calibrated''~\cite{AA11}.
This problem of calibration is generally regarded as `just' a technical step.
Geometrically, we will refer to the \emph{landscape} 
for a given set of detector positions and a given
interferometer as an
$(n-1)$-dimensional surface representing the various coincidence rates in 
those detectors 
as a function of 
the necessary $n-1$ relative arrival times of the~$n$ particles.
Calibration is achieved by knowing the functional form, hence appearance,
of this landscape,
choosing some points on the domain by experimental adjustments
that control arrival times,
and determining the rate at the chosen points.
{When sufficient resolution of the times of arrival are possible,}
doing so over approximately~$n!/2$ points (the exact number is given in Eq.~(\ref{eq:exactnumber}))
in the domain, the rate for simultaneity is then inferred 
(in the perfect case) because enough information is available to shift and rotate the landscape so that the `zero' element of the domain is known. This way of calibrating requires the fewest number of points to estimate the simultaneous-arrival coincidence rate, which is assumed in most treatments of sampling but unfortunately not completely justified in practice.

As \textsc{BosonSampling} is computationally hard,
simulating calibration only adds a smaller computational overhead so it does not make the computational problem any easier.
\textsc{FermionSampling} is another beast altogether.
After calibration is completed, \textsc{FermionSampling} is computationally efficient so calibration is now a key theoretical issue,
not `just' technical.
If a fermion source ejected particles at known, definite times and all fermion paths were calibrated,
then non-simultaneity would be obviated,
but fermion interferometers do not work that way.

We now explain how the input delay channels for the interferometer are 
calibrated~\cite{pi-corrected}.
Intuitively,
consider a source that generates~$n$ particles such that only one particle
is injected into each of~$n$ channels.
Suppose that the source injects all~$n$ particles simultaneously,
but the channel length is unknown so the particles eventually arrive at the~$n$ detectors at uniformly random times despite the promise that all particles are injected simultaneously into each channel.

Calibration is the exercise of adjusting each channel length so that they match, and successful calibration ensures that the simultaneously generated particles are guaranteed to arrive at the detectors simultaneously. If a coincidence-rate model, which depends on arrival times, is trusted, then calibration is achieved by using coincidence-rate data for different choices of input-channel delay increments and fitting coincidence rate data to the model. Then, by interpolation, the appropriate delay-increments can be inferred from the fitted model.
In practice, coincidence rates for simultaneous arrival are extrema so smart search for channel delays that yield a minimum or maximum coincidence rate, as a function of channel delays, can calibrate channel delays without having to resort to model fitting.

For our exact-rate theory, discretized control of delay times between source and interferometer is accommodated here as time bins for arrival.
For our analysis, time-bin width is fixed in all channels.
Thus, we consider discrete arrival times rather than treating arrival time
as a continuum. The number of time bins per channel is~$b$,
which is the same for each channel,
and~$b$ is the ratio of the total run time for the sampling procedure 
to the photodetector arrival time.

We can thus think of the arrival time as a digit in base~$b$.
As~$n$ particles are in play,
the arrival time is a length~$n$ string of digits in base~$b$.
For example, if~$b$ is 16,
a digit could be expressed in hexadecimal,
so an example of an arrival time for particles in 8 channels could be expressed as the length~8 hexadecimal number~12B9B0A1,
which is the hexadecimal representation of~314159265.
Total ignorance about channel delays corresponds to the uniform prior (distribution) over all hexadecimal numbers from~0 to~FFFFFFFF,
which is 4294967295 in decimal.

The trusted model computes the expected coincidence rate for each choice of length-8 hexadecimal strings,
and the calibration task is to adjust each channel's delay so that particle arrivals all have the same digit;
in our length-8 hexadecimal-string example,
the arrivals are all 00000000 or 11111111 or 22222222 and so on up to FFFFFFFF.
All these repeating sequences are equivalent because the model prediction is based only on relative input-channel delays so all digits being the same yield the same predicted coincidence rate.

For calibration, a difficult question concerns how many instances of channel-delay choices must be tested to enable solving of the model parameters.
From Eq.~(\ref{eq:exactnumber}),  we know that the number of samples must scale as~$n!/2$,
and the exact expression is known.
Consequently,   to reduce the overhead required in calibrating the fermionic interferometers deterministic and on-demand single fermion sources are necessary. 

\section{Discussion}
\label{sec:discussion}

In studying this problem,
we have obtained some interesting, instructive results along the way, which we now summarize. Instead of simultaneity,
we showed that computing exact multi-particle coincidence rates,
whether bosonic or fermionic,
incorporates mutual pairwise particle distinguishability,
which can be expressed in terms of mode-overlap fidelity.
This mode-overlap fidelity is what is controlled by channel delays between sources and interferometer,
i.e., by `arrival time'.

Our second claim, obtained section \S\ref{sec:exactrates}, is the exact coincidence-rate mathematical expression for any configuration of arrival times, and for any number of particles.
By configuration,
we refer to an array of relative arrival times for the incoming particles, with empty values~$\emptyset$ for empty input channels. In deriving the coincidence rate equations we introduce the delay matrix, which contains all information on the pairwise levels of distinguishability between particles.

As our third claim, we devised in 
\S\ref{subsec:block} an algorithmic method to simplify
the calculation of coincidence rates, given possibly non-simultaneous arrival times.
For any $s\in G_{m,n}$ our method first computes irreducible matrix 
representations of the submatrix $A(s)$ of a Haar-random 
unitary matrix $U$; we then compute  rates from these 
irreducible matrix representations. For $n$ mutually partially distinguishable particles, our method is 
computationally elegant and advantageous because it is independent of the arrival times and of the numerical entries in $A(s)$, but instead leverages permutation symmetries to block diagonalize an $n! \times n!$ matrix into a sum of smaller blocks, thereby focusing the 
computation on the non-zero entries of the block-diagonal 
form. It is a clear counterpoint to 
the more common approach of using only 
permanents~\cite{shchesnovich2015partial,tamma2016boson,tichy2015sampling}.

Our next claim, presented in \S\ref{subsec:timebin}, is that we use this notion of controllable delays to formulate discretized-time coincidence rates in terms of time bins whose temporal width is the precision of channel time-delay control.
Experimentally, and even philosophically,
time steps are not infinitesimally tunable:
a precision limit exists in practice,
and this practical limit determines our coarse-graining scale to establish time bins.
Throughout this paper we have referred to these time bins as discrete arrival times alluding,
somewhat loosely, to controlled relative (to an arbitrary zero reference) arrival times of particles at the interferometer.  Moreover, for a known set of discrete arrival times $\bm\tau$, we further
established in \S\ref{sec:restricting}
that specific blocks in the general sum are 
automatically $0$ as a result of Gamas's theorem; this leads to considerable 
simplifications when computing coincidence rates for a 
known set of discrete arrival times.

The next claim that we highlight in \S\ref{subsec:witnesscomplexity} is our algorithm, which
shows that our full exact coincidence-rate expression 
contains a configuration instance that requires exponentially many arithmetic operations with respect to total fermion
number~$n$.  
We dubbed the partition associated with this configuration instance
the ``witness partition''.
For~$n$ even, this witness partition is $(n/2,n/2)$ and will appear in experiments when the~$n$ fermions are equally distributed over exactly two distinct time bins, or for any delay partition dominated by $(n/2,n/2)$, i.e. experimentally when at most $n/2$ fermions occupy a single time bin. The group functions or alternatively immanants associated with the witness partition needs to be evaluated and the cost of the computation scales at least exponentially with $n$.

For~$n$ odd, the witness partition is $((n+1)/2,(n-1)/2)$ and appears in experiments where $(n+1)/2$ 
{fermions} arrive in a single time bin, while the remaining $(n-1)/2$ fermions arrive in a single but distinct time bin; again for any delay partition dominated by $((n+1)/2,(n-1)/2)$, i.e. experimentally when at most $(n+1)/2$ fermions occupy a single time bin, the group functions or alternatively immanants for the witness partition needs to be evaluated and the cost of the computation also scales
at least exponentially with $n$, where $n$ is the number of mutually partially distinguishable particles.
 
Our sixth claim is that we calculate the probability for nonzero contribution of the above witness partition 
to the exact coincidence-rate calculation given uniformly 
random arrival times, as is expected when a calibration procedure is required.
We show that this instance occurs with probability going to $1$ as~$n$ 
increases. This is discussed in \S\ref{subsec:p}.

Finally, we have formulated the computational problem of \textsc{GenFermionSampling} in Definition \ref{def:Gfermionsampling}. This formulation is simple but important in that simultaneity is not intrinsic to the definition. By discarding any requirement of simultaneity, boson sampling becomes meaningfully extendable to generalized fermion sampling, and we speculate that hard-to-compute complexity can now arise in both problems.

\section{Conclusions}
\label{sec:conclusions}


Our study provides a more complete understanding of the interference of partially-distinguishable particles, but the motivation for our work is stronger:
our result provides an incentive
to investigate the actual computational complexity of the problem defined
in this work.
Thus far, the focus has been on bosons but fermion sampling could extend the range of possible experiments to reach 
quantum advantage.
Multipartite fermionic interferometry is conceivable in various settings \cite{bauerle2018coherent} such as two-dimensional electron gases,
currently restricted to two-electron interferometry.
Perhaps scaling up fermion interferometry could prove to be more feasible, for many particles,
than for photons, which are the currently favoured particle experimentally.

Finally, our work on the $n$-particle interference of bosons and fermions immediately raises questions about more exotic particles.
How can our methods be applied to coincidence rates for anyons, supersymmetric particles or even strings?
We do not broach these challenging topics here;
rather we think that creating a common foundation for fermions and bosons creates well-posed questions that transcend these types of particles.
\section*{Acknowledgements}
  HdG would like to thank David Amaro Alcal\'{a} and 
Ari Boon for help with 
computer diagonalizations of some low-$n$ cases,
and Olivia Di~Matteo for helpful discussions. DS thanks Andrew Berget for helpful discussions.
BCS and HdG acknowledge financial support from NSERC of Canada.
DS acknowledges financial support from the Ontario Graduate 
Scholarship. 

\bibliography{testbib.bib}   

\end{document}